%% file: EFXmain.tex
\documentclass[11pt]{article}

\usepackage[left=1in, right=1in, top=1in, bottom=1in]{geometry}
\usepackage[group-separator={,}]{siunitx}
\usepackage{amsmath,amsfonts,amssymb, amsthm}
\usepackage{graphicx}
\usepackage{url}
\usepackage{mathtools}
\usepackage[textsize=tiny,textwidth=1cm,shadow]{todonotes}
\usepackage{thmtools}
\usepackage{enumitem}
\usepackage{pdflscape}
\usepackage{verbatim}
\usepackage{nccmath}
\usepackage[utf8]{inputenc}
\usepackage{array}
\definecolor{MyBlue}{rgb}{0.12, 0.12, 0.76}
\usepackage[colorlinks,allcolors=MyBlue]{hyperref}

\usepackage{algorithm}
\usepackage{algpseudocode}
\usepackage{algorithmicx}
\let\oldReturn\Return
\renewcommand{\Return}{\State\oldReturn}
\algtext*{EndWhile}
\algtext*{EndIf}
\algtext*{EndForAll}
\algtext*{EndFor}
\algtext*{EndFunction}

\usepackage{pifont}
\newcommand{\cmark}{\ding{51}}%
\newcommand{\xmark}{\ding{55}}%

\usepackage{pgf}
\usepackage{tikz}
\usetikzlibrary{shadows,arrows,decorations,decorations.shapes,backgrounds,shapes,snakes,automata,fit,petri,shapes.multipart,calc,positioning,shapes.geometric,graphs,graphs.standard}

\makeatletter
\newcommand{\thickhline}{%
    \noalign {\ifnum 0=`}\fi \hrule height 1.4pt
    \futurelet \reserved@a \@xhline
}
\newcolumntype{"}{@{\hskip\tabcolsep\vrule width 1.4pt\hskip\tabcolsep}}
\makeatother

\allowdisplaybreaks[1] 

\newtheorem{theorem}{Theorem}[section]
\newtheorem{lemma}{Lemma}[section]

\newtheorem{definition}{Definition}[section]

\DeclareMathOperator*{\argmax}{arg\,max}
\DeclareMathOperator*{\argmin}{arg\,min}

\newenvironment{proofsketch}{%
  \proof}{\endproof}
\newcommand{\X}{\mathcal{X}}
\newcommand{\Y}{\mathcal{Y}}

\newcommand\citeN[1]{\cite{#1}}
\newcommand\shortcite[1]{\cite{#1}}

\begin{document}

\title{Almost Envy-Freeness with General Valuations}  
\author{Benjamin Plaut \and Tim Roughgarden}
\date{\{bplaut, tim\}@cs.stanford.edu\\ Stanford University}

\maketitle

\input{abstract}

\input{intro}

\input{model}

\input{EFXLowerBound}

\input{localSearchLowerBound}

\input{genIdUpperBounds}

\input{PO}

\input{apxAlgo}

\input{conclusion}

\section*{Acknowledgements}
This research was supported in part by NSF grant CCF-1524062, a Google Faculty Research Award, and a Guggenheim Fellowship.

\bibliographystyle{plain}
\bibliography{refs}

\appendix

\input{additionalProofs}

\input{algoSameRanking}

\end{document}

%% file: abstract.tex
\begin{abstract}
The goal of fair division is to distribute resources among competing players in a ``fair" way. Envy-freeness is the most extensively studied fairness notion in fair division. Envy-free allocations do not always exist with indivisible goods, motivating the study of relaxed versions of envy-freeness. We study the \emph{envy-freeness up to any good} (EFX) property, which states that no player prefers the bundle of another player following the removal of any single good, and prove the first general results about this property. We use the leximin solution to show existence of EFX allocations in several contexts, sometimes in conjunction with Pareto optimality. For two players with valuations obeying a mild assumption, one of these results provides stronger guarantees than the currently deployed algorithm on Spliddit, a popular fair division website. Unfortunately, finding the leximin solution can require exponential time. We show that this is necessary by proving an exponential lower bound on the number of value queries needed to identify an EFX allocation, even for two players with identical valuations. We consider both additive and more general valuations, and our work suggests that there is a rich landscape of problems to explore in the fair division of indivisible goods with different classes of player valuations.
\end{abstract}

%% file: intro.tex
\section{Introduction}\label{sec:intro}

Fair division has a long history, with the earliest known
mechanism for solving the problem dating back to the Bible. No, not
war; the cut-and-choose protocol. When Abraham and Lot first arrive in
the land of Canaan, Abraham suggests that they divide the land between
them. Abraham partitions the land into two parts and lets Lot choose
which part he would like to keep.

What makes this procedure fair? By dividing the land into two pieces
he values equally, Abraham can ensure that he will not envy Lot's
piece, regardless of which piece Lot takes. Since Lot presumably
chooses his favorite piece, he will not envy Abraham. This means that
the cut-and-choose protocol guarantees an \emph{envy-free} allocation,
meaning that each player likes their allocation at least as much as any
other player's allocation.

The cut-and-choose protocol is defined for two players and
\emph{divisible} goods,
meaning that each good can be divided into arbitrarily small pieces. In this paper, we consider the setting of \emph{indivisible} goods,
meaning that the resource in question is a set of discrete goods, each
of which must be wholly allocated to a single player. Unfortunately,
envy-freeness cannot be guaranteed in this setting. We see this even
with two players and a single good: one player must receive the good,
and the other will surely be envious.

Consequently, other notions of fairness are needed. Budish~\shortcite{Budish11:Combinatorial} introduced the concept of \emph{envy-freeness up to one good} (EF1). In an EF1 allocation, player $i$ may envy player $j$, but the envy could be eliminated by removing a single good from player $j$'s allocation. The good is not actually removed; this is a thought experiment used in the definition of envy-freeness up to one good. An EF1 allocation always exists, and can be computed in polynomial time ~\cite{Lipton04:Approximately}.\footnote{The algorithm of \citeN{Lipton04:Approximately} was originally published in 2004 with a different property in mind, as the EF1 property was not proposed until 2011 by \citeN{Budish11:Combinatorial}.}

Caragiannis et al.~\shortcite{Caragiannis16:Nash} proposed another
fairness criterion, one which is strictly stronger than EF1, but
strictly weaker than full envy-freeness. An allocation is
\emph{envy-free up to any good} (EFX) if for any $i,j$ where player $i$
envies player $j$, removing \emph{any} good from $j$'s allocation
would eliminate $i$'s envy. Do EFX allocations always exist? 
This paper takes the first steps toward answering this question.

\subsection{Applications}\label{sec:applications}
The non-profit website Spliddit (www.spliddit.org) is one of the most promising applications of fair division theory~\cite{Goldman15:Spliddit}. Spliddit implements mechanisms for several fair division problems: rent division~\cite{Gal16:Fairest}, taxi fare division, credit assignment (i.e., for a group project or academic paper)~\cite{deClippel08:Impartial}, task distribution~\cite{Pazner78:Egalitarian,Budish13:Designing}, and distribution of indivisible goods. These mechanisms are available for public use at no cost: users can simply log in, define what is to be divided, and enter their valuations. Since the site's launch in November 2014, there have been over 60,000 users~\cite{Caragiannis16:Nash}. The company Fair Outcomes, Inc. (http://www.fairoutcomes.com) offers fair division services in a similar vein.

Another compelling fair division application is allocating courses among students. Students have preferences regarding which courses they would like to take, but each course has a limited capacity. The Wharton School at the University of Pennsylvania now uses a theoretically-grounded mechanism titled Course Match to fairly allocate courses among MBA students, which has led to demonstrably higher satisfaction and perceived fairness among students \cite{Budish11:Combinatorial,Budish16:Course}.

A major selling point of these services is that their solutions are \emph{guaranteed} to satisfy certain fairness properties. For example, in the case of distribution of indivisible goods on Spliddit, users \emph{know} that the solution will be envy-free up to one good and Pareto optimal~\cite{Caragiannis16:Nash}. Our hope is that further work in the area of fair division of indivisible goods will allow user-facing services like Spliddit, Fair Outcomes, Inc., and Course Match to offer users even stronger fairness guarantees.

\subsection{Prior work}\label{sec:prior}

A detailed survey of the fair division literature is outside the scope
of this paper, and we discuss only the works most closely related to
ours.  See e.g.,~\cite{Brams96:Cake,Moulin03:Fair,Brandt16:Handbook}
for further background.

Lipton et al.~\shortcite{Lipton04:Approximately} gave an algorithm whose solution is guaranteed to be EF1 for general valuations. By a {\em valuation}, we mean a function specifying a player's value for each bundle she might receive. By {\em general}, we mean that the only assumptions imposed on valuation functions are normalization (the value of the empty set is 0), and monotonicity (adding goods to a bundle cannot make it worse).

Their algorithm allocates the goods in rounds and ensures that the partial allocation at the end of each round is EF1. At the beginning of each round, an unenvied player is identified; if no such player exists, there must be a cycle of envy, and bundles can be swapped along such cycles until no cycles of envy remain. An arbitrary good is then given to this unenvied player. This player may become envied after receiving this good, but the envy could be eliminated by removing the good she just received (since she was unenvied prior to receiving that good). This ensures that whenever player $i$ envies player $j$, the envy could be eliminated by removing the most recent good given to player $j$, so the resulting allocation is EF1.

Caragiannis et al.~\shortcite{Caragiannis16:Nash} studied the case where valuations are additive, meaning that each player's value for a set of goods is the sum of her values for the individual goods. They showed that the allocation maximizing the product of players' utilities (the maximum Nash welfare solution) is guaranteed to be both EF1 and Pareto optimal, assuming valuations are additive. In contrast, the algorithm of Lipton et al.~\shortcite{Lipton04:Approximately} does not guarantee a Pareto optimal allocation.\footnote{Suppose there are two players with additive valuations $v_1, v_2$ over three goods, $a, b, c$, where $v_1(\{a\}) = 3, v_1(\{b\}) = 2, v_1(\{c\}) = 4$ and $v_2(\{a\}) = 4, v_2(\{b\}) = 3, v_2(\{c\}) = 2$. The algorithm of Lipton et al.~\shortcite{Lipton04:Approximately} could first allocate $a$ to player 1, then $b$ to player 2, and finally $c$ also to player 2. The resulting allocation is EF1, but giving $\{c\}$ to player 1 and $\{a, b\}$ to player 2 would be better for both players.}

Caragiannis et al.~\shortcite{Caragiannis16:Nash} also proposed the fairness criterion of envy-freeness up to any good, and left the possible existence of EFX allocations as an open problem. We are not aware of any results regarding EFX allocations prior to this work.

We briefly describe several other models for fair division of indivisible goods. Brams et al.~\shortcite{Brams17:Maximin} and Aziz et al.~\shortcite{Aziz15:Fair} assumed that players express only an ordinal ranking over the goods, as opposed to exact values. Certain tasks become easier in this domain, but important information is arguably lost by only considering rankings. Randomized allocations have also been considered (e.g., \citeN{Bogomolnaia04:Random, Budish13:Designing}), but this is not suitable for the applications we are most interested in, where the outcome is only used once. Dickerson et al.~\shortcite{Dickerson14:Computational} took a probabilistic approach, and showed that envy-free allocations are likely to exist when the number of goods is at least a logarithmic factor larger than the number of players. While illuminating, this does not directly bear on our goal: determining when fair allocations are guaranteed to exist, and how they can be computed.

\subsection{Our contributions}

\renewcommand{\arraystretch}{1.6}
\begin{table*}
\resizebox{6.5 in}{!}{ 
\begin{tabular}{ c"c|c|c|c|c} 
  & $n=2$, add & $n=2$, gen & $n\geq2$, gen + id & $n>2$, add & $n>2$, gen \\ 
  \thickhline
 $\frac{1}{2}$ EFX & \cmark $\ $ (Thm.~\ref{thm:two-players}) 
 		& \cmark $\ $ (Thm.~\ref{thm:two-players}) 
		& \cmark $\ $ (Thm.~\ref{thm:general-identical}) 
		& \cmark $\ $(Thm.~\ref{thm:apx-upper}) & \textbf{?} \\
 \hline
 EFX & \cmark $\ $ (Thm.~\ref{thm:two-players}) & \cmark $\ $(Thm.~\ref{thm:two-players})
 	 & \cmark $\ $(Thm.~\ref{thm:general-identical}) & \textbf{?} & \textbf{?}\\
 \hline
 EFX + PO (nmu) & \cmark $\ $(Thm.~\ref{thm:two-players-po}) 
 	& \xmark $\ $(Thm.~\ref{thm:po-counterexample})
 	& \cmark $\ $(Thm.~\ref{thm:gen-id-po}) & \textbf{?} & \xmark $\ $ (Thm.~\ref{thm:po-counterexample})\\ 
\end{tabular}
}
 \vspace{.15 in}
\caption{A summary of our existence results. Here $n$ is the number of players. ``add", ``gen", ``id", and ``nmu" refer to additive valuations, general valuations, identical valuations, and nonzero marginal utility, respectively. ``\cmark" indicates that the type of allocation specified by the row is guaranteed to exist in the setting specified by the column, while ``\xmark" indicates that we give a counterexample, and ``\textbf{?}" indicates an open question.}
\label{tbl:results-table}
\end{table*}
\renewcommand{\arraystretch}{1}

We consider the EFX property in a variety of contexts; our main existence results are given in Table~\ref{tbl:results-table}. 

\subsubsection{Exponential query complexity lower bound}

Section~\ref{sec:lower-bound} presents our most technically involved result: an exponential lower bound
on the number of value queries required by a deterministic algorithm
to find an EFX allocation.
This is done via a reduction from local
search on a class of graphs known as the Odd graphs, for which we
prove an exponential lower bound. In combination with results due to Dinh and Russell~\shortcite{Dinh10:Quantum} and
Valencia-Pabon and Vera~\shortcite{Valencia-Pabon05:Diameter}, this yields an analogous
exponential lower bound for randomized algorithms. Dobzinski et al.~\shortcite{Dobzinski15:Complexity} also use a local search reduction to prove a lower bound on the number of value queries required to find a certain type of equilibrium in a simultaneous second price auction, for bidders with XOS (i.e., fractionally subadditive) valuations. We hope that this lower bound technique will be useful in other contexts as well.

Our lower bound holds even for two players with identical submodular
valuations. In stark contrast, the algorithm of Lipton et
al.~\shortcite{Lipton04:Approximately} finds an EF1 allocation in
polynomial time for general and possibly distinct valuations, and for
any number of players. This suggests that EFX is indeed a
significantly stronger fairness guarantee than EF1, and deserves
further study.

\subsubsection{Positive EFX results}
Many of our positive results rely on the leximin solution. The leximin (a portmanteau of ``lexicographic" and ``maximin") solution selects the allocation which maximizes the minimum utility; then, if there are multiple allocations which achieve that minimum utility, it chooses among those the allocation which maximizes the second minimum utility, and so on. The leximin solution was developed as a metric of fairness in and of itself~\cite{Rawls71:Theory,Sen76:Welfare,Sen77:Social}, and has been used before in fair division, though typically for randomized allocations (e.g.~\cite{Bogomolnaia04:Random}).

In Section~\ref{sec:gen-id}, we show that when players have general but identical valuations, a modification of the leximin solution is EFX. By {\em identical valuations}, we mean that all players have the same valuation. This result also yields a cut-and-choose-based protocol for two players with general and possibly distinct valuations that is guaranteed to produce an EFX allocation. This is consistent with our exponential lower bound, as it is well known that finding the leximin solution can require exponential time for general valuations (e.g. ~\cite{Dobzinski13:Communication}).\footnote{We mention that Appendix Section~\ref{sec:same-ranking} shows that for two players with additive valuations, an EFX allocation can be computed in polynomial time by a different method.}

These positive results
contrast with the state-of-the-art for possibly distinct valuations and three or more players,
where even for additive valuations, 
the guaranteed existence of an EFX allocation remains an open question
(``despite significant effort,'' according to
\citeN{Caragiannis16:Nash}).

\subsubsection{EFX and Pareto optimality}

In Section~\ref{sec:po}, we consider Pareto optimality. In economics, an outcome is Pareto optimal (PO) if there is no way to make one player better off without making another player worse off. We show that even in simple cases, it is possible that no EFX allocation is also PO. However, these cases rely on a player having zero value for a good being added to her bundle. 

We propose the assumption that adding a good to a player's bundle strictly improves the player's value for that bundle, and refer to this as ``nonzero marginal utility". We view this as quite a weak assumption: especially in real-world settings, one might expect a player to always prefer to have a good than not.

Under this assumption, we show that for two players with additive valuations, the leximin solution is both EFX and PO.\footnote{When discussing the leximin solution for players with different valuations, we assume that each player's value for the entire set of goods is the same: were this not true, we could simply rescale the valuations as needed and use the leximin solution over the rescaled valuations.}  We also show that for any number of players with general but identical valuations, the leximin solution is EFX and PO. Finally, we give a counterexample where for two players with distinct general valuations, no EFX allocation is PO (even assuming nonzero marginal utility).

\begin{figure}
\centering
\begin{tabular}{ c|ccc} 
& a & b & c\\
\hline
player 1 & 5 & 3 & 1 \\
player 2 & 5 & 1 & 3\\
\end{tabular}
\caption{An instance where our algorithm provides stronger guarantees than the algorithm currently deployed on Spliddit. Here two players have additive valuations over three goods, $a,b,$ and $c$. By symmetry, assume $a$ is given to player 1. Spliddit selects the maximum Nash welfare solution, which gives $\{a,b\}$ to player 1 and $\{c\}$ to player 2. This is EF1 and PO, but not EFX, since player 2 would still envy player 1 after the removal of $b$. Our algorithm returns the unique (up to symmetry) EFX and PO allocation, which gives $\{a\}$ to player 1 and $\{b,c\}$ to player 2.} 
\label{fig:better-than-spliddit}
\end{figure}

\subsubsection{Comparison to Spliddit in the two player case}\label{sec:outperforming}
Perhaps of most practical importance is our result that for two players with additive valuations and nonzero marginal utility, the leximin solution is both EFX and PO. This provides stronger guarantees than the currently deployed algorithm on Spliddit, which selects the maximum Nash welfare solution, and only guarantees an allocation which is EF1 and PO.\footnote{Spliddit only considers additive valuations. This is because each user need only report $m$ values to specify an entire additive valuation; in contrast, an exponential number of values can be required to specify a general valuation.}


This manifests even in simple examples, such as the instance given by Figure~\ref{fig:better-than-spliddit}. By symmetry, assume that player 1 receives good $a$. The maximum Nash welfare solution selects the allocation which maximizes the product of utilities: in this case, that would give player 1 $a$ and $b$, and player 2 only $c$. This allocation is EF1, because player 2 would not envy player 1 if good $a$ were removed from player 1's bundle. However, the allocation is not EFX, because player 2 \emph{would} envy player 1 even if good $b$ were removed from player 1's bundle. 

In contrast, our algorithm returns the unique (up to symmetry) EFX and PO allocation, which gives $a$ to player 1 and $b$ and $c$ to player 2. We suggest that this is also the more intuitively fair allocation. Furthermore, the assumption of nonzero marginal utility seems especially reasonable in the case of two players with additive valuations: if a player is truly indifferent to some good, one could imagine simply giving that good to the other player and excluding it from the fair division process entirely.\footnote{A vindictive player might object to this: she may be unhappy with the other player receiving a good ``for free", even if she has zero value for the good herself. We argue that this constitutes having nonzero value for the good, and that a player has zero value for a good only if she is truly indifferent.} We do note that Spliddit's current algorithm does not require the assumption of nonzero marginal utility, however. Neither approach has a clear advantage in terms of computational efficiency: both the leximin solution and maximum Nash welfare solution are NP-hard to compute, even for two players with additive valuations.\footnote{For two players with identical additive valuations, the leximin solution gives each player half the total value if and only if the valuation is a ``yes" instance of the partition problem. The reduction is less obvious for maximum Nash welfare; see e.g. \cite{Ramezani10:Nash}.}

Finally, in Section~\ref{sec:apx-algo} we propose an approximate version of EFX, and show that a $\frac{1}{2}$-EFX allocation always exists when players have subadditive (possibly distinct) valuations.



More broadly, our results span additive, submodular, subadditive, and
general valuations, and identify separations between these classes
from a fair division perspective. For example, we show that assuming
nonzero marginal utility and two players with additive valuations, an
allocation which is both EFX and PO is guaranteed to exist, while
there is a counterexample for two players with
general valuations.
Such valuation
classes have already played a central role in the development of
algorithmic mechanism design over the past 15 years
(e.g.~\cite{Lehmann01:Auctions}), and they may well prove equally
important in the fair division of indivisible goods.



%% file: model.tex
\section{Model}\label{sec:model}

Let $[k]$ denote the set $\{1, 2, ...k\}$. Let $N = [n]$ be the set of
players and $M$ be the set of goods, where $m = |M|$. We assume throughout the
paper that goods are indivisible: a good may not be split among
multiple players. Each player $i$ has a value for each subset of $M$,
specified as a \emph{valuation function}
$v_i: 2^M \to \mathbb{R}_{\geq 0}$. 
Throughout the paper, we assume 
normalization, meaning that $v_i(\emptyset) = 0$,
and monotonicity (a.k.a.\ ``free disposal"), meaning that
$v_i(S) \le v_i(T)$ whenever $S \subseteq T$.
When we refer to ``general valuations,'' 
we mean the set of all valuations that satisfy these two properties.

A special type of valuation is an {\em additive} valuation, where
$v_i(S) = \sum_{g \in S} v_i(\{g\})$ for every $S \subseteq M$. Thus $m$
parameters (one for each good) implicitly specify the $2^m$ values of the valuation. The majority of the literature on computational fair division, with both
divisible and indivisible goods, focuses on additive valuations.  There
are also many interesting subclasses of valuations that generalize
additive valuations.  For example, our main lower bound result
(Theorem~\ref{thm:efx-complexity}) holds for {\em submodular}
valuations, which are valuation functions~$v$ that satisfy
``diminishing returns'':
\[
v(S \cup \{x\}) - v(S) \geq v(T \cup \{x\}) - v(T) 
\]
for every $S \subseteq T$ and $x \notin T$. One of our positive results, Theorem~\ref{thm:apx-upper}, will hold for \emph{subadditive} valuations. A valuation $v$ is subadditive if
\[
v(S) + v(T) \geq v(S \cup T)
\]
Every additive valuation is submodular, and every submodular valuation is subadditive.


An \emph{allocation} $A$ is a partition of $M$ into $n$ disjoint
subsets, $(A_1, A_2...A_n)$, where $A_i$ is 
the {\em bundle} given to player $i$. 
We refer to
an allocation as \emph{partial} if only a subset $S\subseteq M$ of the
goods are allocated. When ``partial'' is omitted, it means that
all goods have been allocated.

Our objective is to find a ``fair'' allocation. Many different notions
of fairness have been studied, with envy-freeness being one of the
most prominent (see e.g.,~\cite{Brams96:Cake,Moulin03:Fair,Brandt16:Handbook}
for further background).

\begin{definition}
An allocation $A$ is \emph{envy-free} if for all $i$ and $j$,
\begin{align*}
v_i(A_i) \geq v_i(A_j).
\end{align*}
\end{definition}

We say that {\em $i$ envies $j$} if $v_i(A_i) < v_i(A_j)$. Unfortunately, an envy-free allocation does not always exist in the context of indivisible goods. This is clear even with two players and one good: the player who does not receive the good will envy the other, assuming they both have nonzero value for the good. Furthermore, determining whether an envy-free allocation exists is NP-complete~\cite{Bouveret08:Efficiency}: with two players and identical additive valuations, this is the partition problem.

Consequently, a relaxed version of envy-freeness has been studied, called envy-freeness up to one good~\cite{Budish11:Combinatorial,Caragiannis16:Nash}.

\begin{definition}
An allocation $A$ is \emph{envy-free up to one good} (EF1) if for all $i,j$ where $i$ envies $j$,\footnote{The ``where $i$ envies $j$" clause is necessary, or the condition would technically fail when $A_j = \emptyset$. This is not an issue for the definition of EFX, however.}
\begin{align*}
\exists\ g\in A_j\ \text{such that}\ v_i(A_i) \geq v_i(A_j \backslash \{g\}).
\end{align*}
\end{definition}

That is, $i$ may envy $j$, but there is a good in $j$'s bundle such
that if it were removed, $i$ would no longer envy $j$. An EF1 allocation always exists, and can be computed in polynomial time, even for general valuations~\cite{Lipton04:Approximately}. 

Furthermore, Caragiannis et al.~\shortcite{Caragiannis16:Nash} showed
that for additive valuations, an allocation which is both EF1 and
Pareto optimal always exists; in particular, the maximum Nash welfare
solution \cite{Nash50:Bargaining,Kaneko79:Nash,Ramezani10:Nash}
is guaranteed to satisfy both properties. Caragiannis et al.~\shortcite{Caragiannis16:Nash} also proposed a new fairness notion, one which is strictly weaker than envy-freeness, but strictly stronger than EF1.

\begin{definition}
An allocation $A$ is \emph{envy-free up to any good} (EFX) if, for all $i,j$,
\begin{align*}
\forall g\in A_j,\ \ v_i(A_i) \geq v_i(A_j \backslash \{g\}).
\end{align*}
\end{definition}
In words, EFX demands that removing \emph{any} good from $j$'s bundle
would guarantee that $i$ does not envy $j$. Next, we define the standard notion of Pareto optimality.
\begin{definition}
An allocation $A$ is Pareto optimal if there is no other allocation $B$ where
\begin{align*}
\forall i \in [n],\ \ v_i(B_i) \geq&\ v_i(A_i),\ \ \textnormal{and}\\
\exists j \in [n],\ \ v_j(B_j) >&\ v_j(A_j)
\end{align*}
\end{definition}

Finally, we define an approximate version of EFX. In Section~\ref{sec:apx-algo}, we will give an algorithm which produces a $\frac{1}{2}$-EFX allocation for any number of players with subadditive valuations.
\begin{definition}
An allocation $A$ is $c$-EFX if for all $i,j$,
\begin{align*}
\forall g\in A_j,\ \ v_i(A_i) \geq c \cdot v_i(A_j \backslash \{g\})
\end{align*}
where $0 \leq c \leq 1$.
\end{definition}

%% file: EFXLowerBound.tex
\section{Query complexity lower bound}\label{sec:lower-bound}

We begin with our most technically involved result: an exponential lower bound on the number of value queries required by any deterministic algorithm to compute an EFX allocation. Our lower bound will hold even for two players, and
even if their valuations are restricted to be identical and
submodular.\footnote{There is of course no hope for an exponential
  lower bound for additive valuations, since $m$ value queries
  suffice to reconstruct an entire additive valuation.}

In Section~\ref{sec:local-search-problem}, we introduce the local search problem that we will reduce from. In Section~\ref{sec:efx-complexity}, we prove that finding an EFX allocation is at least as hard as solving local search on a particular class of graphs. In Section~\ref{sec:local-search-complexity}, we show that any deterministic algorithm which finds a local maximum on this class of graphs requires an exponential number of queries. This will imply that the problem of finding an EFX allocation has exponential query complexity as well. Finally, in Section~\ref{sec:rand-complexity}, we extend this lower bound to randomized algorithms.

\subsection{Local search}\label{sec:local-search-problem}

The \textsc{Local Search} problem takes as input an undirected graph
$G = (V, E)$ and an oracle function $f: V \to \mathbb{R}$. The goal
is to find a local maximum $a\in V$, where $f(a) \geq f(b)$ for all
$(a,b) \in E$. Since there exists a global maximum, there must exist
at least one local maximum. We are interested in the number of queries
required to find a local maximum, where a query to $a \in V$ returns
$f(a)$. Queries are the only method by which an algorithm can discover
information about $f$ (i.e., it is given as a ``black box'').
All other operations are free in this model---we count only the number of
queries. Queries can be adaptive, with an algorithm's choice of which
vertex to query next depending on the results of previous queries.

For a graph $G$, the \emph{deterministic query complexity} of
\textsc{Local Search} on $G$ is the minimum number of queries required
by any deterministic algorithm to solve \textsc{Local Search} on
$G$ (for a worst-case choice of~$f$).
Formally, let $D[LS(G)]$ be the deterministic query complexity of
\textsc{Local Search} on $G$. Then $D[LS(G)] = \min\limits_\Gamma \max\limits_f T_{LS}(G, f, \Gamma)$,
where the minimizer ranges over all deterministic algorithms $\Gamma$,
the maximizer ranges over all functions $f: V \to \mathbb{R}$, and
$T_{LS}(G, f, \Gamma)$ is the number of queries used by the algorithm
$\Gamma$ to find a local maximum of $f$ on $G$.

The difficulty of local search depends on the graph $G$. The \emph{Kneser graph} $K(n, k)$ is the graph whose vertices are the
size $k$ subsets of $[n]$, where two vertices are adjacent if and only
if their corresponding subsets are disjoint. The star of our lower bound argument is the {\em Odd graph}, $K(2k + 1, k)$. The most famous Odd
graph is the Petersen graph (Figure~\ref{fig:efx}).

\subsection{Local search on $K(2k+1, k)$ reduces to finding an EFX allocation}\label{sec:efx-complexity}

The \textsc{EFX Allocation} problem takes as input the set of players
$N = [n]$, the set of goods $M = [m]$, and a list of valuations
$(v_1, v_2,\ldots,v_n)$. In general, the goal is to find an EFX
allocation, or determine that none exists. 
The only method by which an algorithm can discover
information about the $v_i$'s is through {\em value queries},
where upon querying the valuation $v_i$ at the set $S$, the algorithm
learns $v_i(S)$.
Our lower bound applies even for a version of the problem that we will show to be total (Theorem~\ref{thm:general-identical}), meaning that
an EFX allocation is guaranteed to exist.

Consider the special case of the \textsc{EFX Allocation} problem
where all valuations are identical.
We will show in Section~\ref{sec:gen-id} that an EFX allocation is
guaranteed to exist in this setting.
We can define the deterministic query complexity $D[EFX_{id}(n,m)]$ as
the minimum number of queries required to find an EFX allocation for a
set of players $N =[n]$ and a set of goods $M=[m]$, given a single
valuation $v$ where an EFX allocation is known to exist. Formally, $D[EFX_{id}(n,m)] = \min\limits_\Gamma \max\limits_v T_{EFX}(N, M, v,
\Gamma)$,
where $T_{EFX}(N,M,v,\Gamma)$ denotes the number of queries required by the
algorithm $\Gamma$ to find an EFX allocation for players $N$ with valuation $v$ over goods
$M$. Since this is a special case of the general \textsc{EFX Allocation} problem, the deterministic query complexity of the general \textsc{EFX allocation} problem is at least  $D[EFX_{id}(n,m)]$.

We now state and prove our main result of Section~\ref{sec:efx-complexity}. We will use $M = [2k+1]$ for some integer $k$.

\begin{samepage}
\begin{theorem}\label{thm:efx-reduction} 
The deterministic query complexity of the \textsc{EFX Allocation}
problem satisfies
\begin{align*}
D[EFX_{id}(2,2k+1)] \geq D[LS(K(2k+1, k))],
\end{align*}
even for two players with identical submodular valuations.
\end{theorem}
\end{samepage}

\begin{proof}
  Let
  $T = D[EFX_{id}(2,2k+1)]$; then there exists an algorithm $\Gamma$
  for finding an EFX allocation which uses at most $T$ queries,
  regardless of $v$. We will construct an algorithm $\Gamma'$ for
  \textsc{Local Search} which also uses at most $T$ queries,
  regardless of $f$. Formally,
  $\max\limits_v T_{EFX}(\{1,2\}, M, v, \Gamma) = T$, and we will
  construct $\Gamma'$ such that
  $\max\limits_f T_{LS}(K(2k+1, k), f, \Gamma') \leq T$.

Define the algorithm $\Gamma'$ on input $(K(2k+1, k), f)$ as follows. For each $S \subseteq [2k+1]$, define $v(S)$ as
\[ v(S) =
   \begin{cases} 
      2|S| & \textnormal{if}\ \ |S| < k\\
      2k - \dfrac{1}{1 + e^{(f(S))}} & \textnormal{if}\ \ |S| = k\\
      2k  & \textnormal{if}\ \ |S| > k.
   \end{cases}
\]
Then run $\Gamma$ on $(\{1,2\}, [2k+1], v)$ to obtain an EFX
allocation $(A_1, A_2)$, and return $A_1$ if $|A_1| < |A_2|$ and
$A_2$ otherwise. We will show that the returned set corresponds to a
local maximum in $K(2k+1, k)$ (see Figure~\ref{fig:efx}).

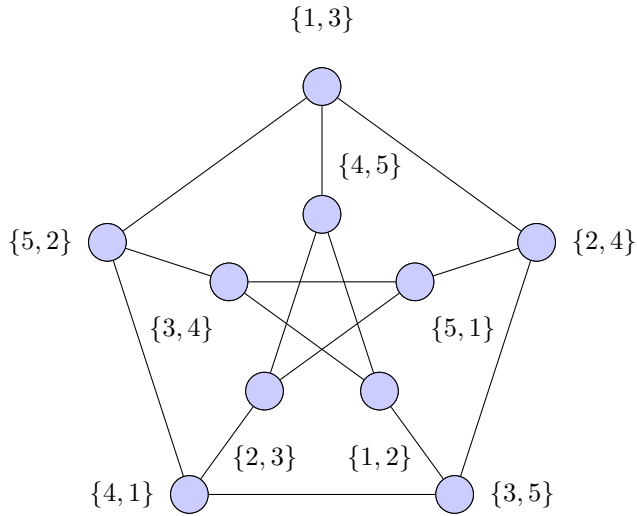
\begin{figure}[ht!bp]
\centering
\begin{tikzpicture}[every node/.style={draw,circle,fill=blue!20,minimum size=5 mm,font=\sffamily\small\bfseries}]
  \graph[clockwise=5, radius=3cm,empty nodes] {subgraph C_n [n=5, name=A],
  a[label=above:{$\{1,3\}$}], b[label=right:{$\{2,4\}$}], c[label=right:{$\{3,5\}$}],
  d[label=left:{$\{4,1\}$}], e[label=left:{$\{5,2\}$}];};
  \graph[clockwise=5, radius=1.3cm, empty nodes] {subgraph I_n [n=5,name=B],
  a[label=above right:{$\{4,5\}$}], b[label=below right:{$\{5,1\}$}], c[label=below:{$\{1,2\}$}],
  d[label=below:{$\{2,3\}$}], e[label=below left:{$\{3,4\}$}];};  
  
  \foreach \i in {1,2,3,4,5}{\draw (A \i) -- (B \i);}
  \newcounter{k}
  \foreach \i in {1,2,3,4,5}{%
  \pgfmathsetcounter{k}{ifthenelse(mod(\i+2,5),mod(\i+2,5),5)}
  \draw (B \i) -- (B \thek);
}
\end{tikzpicture}

\caption{The graph shown is $K(2k+1, k)$ for $k=2$, also known as the Petersen graph. Each vertex corresponds to a size $2$ subset of $[5]$.
Suppose the allocation where
 $A_1 = \{1,2,3\}$ and $A_2 =\{4,5\}$ is EFX. Since
 $v(S_1) \geq v(S_2)$ if and only if $f(S_1) \geq f(S_2)$, we have
 $f(\{4,5\}) \geq f(\{1,2\})$, $f(\{4,5\})\geq f(\{2,3\})$, and
 $f(\{4,5\})\geq f(\{1,3\})$. Therefore \{4,5\} is a local maximum in this graph.
}
\label{fig:efx}
\end{figure}

For brevity, define 
\[
\delta(S) = - \frac{1}{1 + e^{f(S)}}.
\] 
We note
that $-1 < \delta(S) < 0$ for all $S$, and that $\delta(S)$ is
strictly increasing with $f(S)$. Any other function satisfying these
properties would work as well.

We first argue that any EFX allocation returned by $\Gamma$ must give
one player exactly $k$ goods. Suppose that this were not the
case. Then one player must receive fewer than $k$ goods; without loss
of generality, assume $|A_2| < k$, and thus $|A_1| > k+1$. Therefore
$v(A_2) \leq 2k - 2$ and $v(A_1) = 2k$.

For an arbitrary $g\in A_1$, we have $|A_1 \backslash \{g\}| >
k$.
Therefore there exists a $g\in A_1$ such that
$v(A_1 \backslash \{g\}) = 2k > v(A_2)$, so the allocation cannot be
EFX. Thus any EFX allocation must give one player exactly $k$
goods. Therefore $\Gamma$ will return a set of size $k$, which
corresponds to a valid vertex of $K(2k+1, k)$.

Without loss of generality, assume $|A_1| = k+1$ and $|A_2| = k$. Then
$v(A_1) = 2k$ and $v(A_2) = 2k + \delta(A_2) < 2k$, so
$v(A_1) > v(A_2)$. Therefore the allocation $A = (A_1, A_2)$ is EFX if
and only if $v(A_2) \geq v(A_1 \backslash \{g\})$ for all $g\in A_1$.

We can rewrite this condition as $v(A_2) \geq v(S)$ for all
$S\subseteq A_1$ where $|S| = k$. For any $|S| = k$, we have
$v(A_2) - v(S) = \delta(A_2) - \delta(S)$. Since $\delta$ is
strictly increasing with $f(S)$, we have $v(A_2) \geq v(S)$ if and
only if $f(A_2) \geq f(S)$. Therefore an allocation $(A_1, A_2)$ is
EFX if and only if $f(A_2) \geq f(S)$ for all $S\subseteq A_1$ where
$|S| = k$.

Observe that $S\subseteq A_1$ if and only if $S \cap A_2 =
\emptyset$.
Therefore an allocation $(A_1, A_2)$ is EFX if and only if
$f(A_2) \geq f(S)$ for all $S\subseteq M$ where $|S| = k$ and
$S \cap A_2 = \emptyset$. This is exactly the definition of $A_2$
being a local maximum in $K(2k+1, k)$. Therefore an allocation
$(A_1, A_2)$ is EFX if and only if $A_2$ is a local maximum in
$K(2k+1, k)$.

Thus $\Gamma'$ correctly solves \textsc{Local
  Search}. Furthermore, since $\Gamma'$ uses no queries outside of
running $\Gamma$, and $\Gamma$ uses at most $T$ queries, $\Gamma'$
also uses at most $T$ queries. Therefore
\[
D[EFX_{id}(2,2k+1)] \geq D[LS(K(2k+1, k))].
\]

It remains to show that $v$ is submodular. For any $S\subseteq M$ and $x\in M\backslash S$, we have
\[ v(S\cup\{x\}) - v(S) =
   \begin{cases} 
      2 & \textnormal{if}\ \ |S\cup\{x\}| < k\\
      2 + \delta(S\cup\{x\}) & \textnormal{if}\ \ |S\cup\{x\}| = k\\
      - \delta(S) & \textnormal{if}\ \ |S\cup\{x\}| = k + 1\\
      0 & \textnormal{if}\ \ |S\cup\{x\}| > k + 1.
   \end{cases}
\]

Therefore $v(S\cup\{x\}) - v(S)$ is non-increasing with $|S|$, since $-1 < \delta(S) < 0$ for all $S$. Thus $v(X\cup\{x\}) - v(X) \geq v(Y\cup\{x\}) - v(Y)$ whenever $|X| < |Y|$. If $X \subseteq Y$, either $|X| < |Y|$ or $X = Y$. When $X = Y$, we trivially have $v(X\cup\{x\}) - v(X) = v(Y\cup\{x\}) - v(Y)$. Thus we have $v(X\cup\{x\}) - v(X) \geq v(Y\cup\{x\}) - v(Y)$ whenever $X\subseteq Y$, and so $v$ is submodular.
\end{proof}

%% file: localSearchLowerBound.tex
\subsection{Query complexity of local search on Odd graphs}\label{sec:local-search-complexity}

In this section, we show that finding a local maximum on $K(2k+1, k)$ has exponential query complexity, completing our lower bound on the number of queries required to find an EFX allocation.\footnote{A similar lower bound for local search on $K(2k+1, k)$ was proved (using different arguments) in \cite{Dobzinski15:Complexity}.}

\subsubsection{The role of boundaries}

For a graph $G=(V,E)$ and a set $S \subseteq V$, define the
\emph{boundary} $B(S)$ of $S$ as the set of vertices that are not in
$S$ but are adjacent to a vertex in $S$. Formally,
$B(S) = \{a \in V\backslash S: \exists b\in S, (a,b) \in E\}$.
The next result, due to \citeN{Llewellyn89:Local},
implies that local search is hard in graphs that only
have large boundaries.

\begin{lemma}[\cite{Llewellyn89:Local}]\label{lem:sep-game}
For any graph $G = (V,E)$ and integers $t$ and $c$,
\begin{align*}
D[LS(G)] \geq \min\big(t, \min_{S} \{|B(S)|: c -t \leq |S| \leq c\}\big).
\end{align*}
\end{lemma}

\begin{proofsketch}
  We sketch a proof for the benefit of the
  reader. The proof follows an adversary argument. 
Let $G_u$ be the
  subgraph induced by the still-unqueried vertices. While $G_u$
  remains connected, suppose
  the adversary simply returns increasing values for each query. Then
  the only way for a local maximum to be created is to query a vertex
  $a$ after querying all of $a$'s neighbors. 

Furthermore, while $G_u$ remains connected and contains at least one unqueried vertex, the most
recently queried vertex $a$ must have an unqueried neighbor $b$: if
not, $G_u$ must have been disconnected prior to the most recent
query. The adversary is free to toggle which of $a$ and $b$ is a local
maximum (or possibly neither, if there are more unqueried vertices). 
Thus while at least one vertex has not been queried and the graph of unqueried vertices remains connected, it cannot be determined where the graph has a local maximum.

Thus the only strategy to counteract the adversary is to perform a sort of binary search. First, we must disconnect the graph of unqueried vertices. At least one of the resulting components must contain a local maximum, and Llewellyn et al.~\shortcite{Llewellyn89:Local} show how we can always identify one such component based on the query results so far. Thus we can recurse on that component, and the process repeats. Llewellyn et al.~\shortcite{Llewellyn89:Local} call this the \emph{separation game}. An example of the separation game being played on a path is given by Figure~\ref{fig:sep-game}.

By this logic, we will have to eventually disconnect a ``fairly large" component: if it is too small, the adversary is free to place the local maximum in another larger component. Specifically, Llewellyn et al.~\shortcite{Llewellyn89:Local} show that for any integers $t$ and $c$, the adversary can force us to either query $t$ vertices, or disconnect a set of vertices $S$ where $c -t \leq |S| \leq c$.

In order to disconnect a set of vertices $S$, every vertex on the
boundary of $S$ must be queried. Thus at least $\min\big(t,
\min\limits_{S} \{|B(S)|: c -t \leq |S| \leq c\}\big)$ must be
queried, as claimed.
\end{proofsketch}

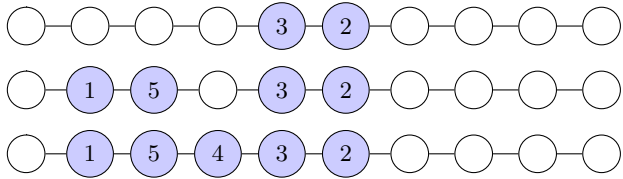
\begin{figure}[ht!bp]
\centering
\begin{tikzpicture}[every node/.style={draw,circle,minimum size=5 mm, node distance=.85 cm, font=\footnotesize}]
  \node (1) {$$};
  \node (2) [right of=1] {$$};
  \node (3) [right of=2] {$$};
  \node (4) [right of=3] {$$};
  \node [fill=blue!20] (5) [right of=4] {$3$};
  \node [fill=blue!20] (6) [right of=5] {$2$}; 
  \node (7) [right of=6] {$$}; 
  \node (8) [right of=7] {$$};
  \node (9) [right of=8] {$$}; 
  \node (10) [right of=9] {$$}; 
  \node (11) [below of=1] {$$};
  \node [fill=blue!20] (12) [right of=11] {$1$};
  \node [fill=blue!20] (13) [right of=12] {$5$};
  \node (14) [right of=13] {$$};
  \node [fill=blue!20] (15) [right of=14] {$3$};
  \node [fill=blue!20] (16) [right of=15] {$2$}; 
  \node (17) [right of=16] {$$}; 
  \node (18) [right of=17] {$$};
  \node (19) [right of=18] {$$}; 
  \node (20) [right of=19] {$$}; 
  \node (21) [below of=11]{$$};
  \node [fill=blue!20] (22) [right of=21] {$1$};
  \node [fill=blue!20] (23) [right of=22] {$5$};
  \node [fill=blue!20] (24) [right of=23] {$4$};
  \node [fill=blue!20] (25) [right of=24] {$3$};
  \node [fill=blue!20] (26) [right of=25] {$2$}; 
  \node (27) [right of=26] {$$}; 
  \node (28) [right of=27] {$$};
  \node (29) [right of=28] {$$}; 
  \node (30) [right of=29] {$$}; 
  
  \foreach \i [remember=\i as \lasti (initially 1)] in {1,2,3,4,5,6,7,8,9,10} {\draw (\lasti) -- (\i);};\
  \foreach \i [remember=\i as \lasti (initially 11)] in {11,12,13,14,15,16,17,18,19,20} {\draw (\lasti) -- (\i);};
  \foreach \i [remember=\i as \lasti (initially 21)] in {21,22,23,24,25,26,27,28,29,30} {\draw (\lasti) -- (\i);};
  
\end{tikzpicture}
\caption{An example of the separation game played on a path. After two central vertices are queried, returning values $3$ and $2$ as shown, we know that there must be a local maximum in the left half. Next, we bisect the left half by querying two more vertices, which return values $1$ and $5$. At this point, we know that either the vertex with value $5$ or the vertex immediately to its right must be a local maximum, and only one more query is required to determine which. In this case, the local maximum is the vertex with value $5$.}\label{fig:sep-game}
\end{figure}

\subsubsection{Boundaries of Kneser graphs}
In light of Lemma~\ref{lem:sep-game} and our interest in Kneser
graphs, the natural next step it to understand boundary sizes in
Kneser graphs.  The next lemma is due to Zheng~\shortcite{Zheng15:Vertex}.

\begin{lemma}[\cite{Zheng15:Vertex}]\label{lem:boundary} Let 
$\mu_G(r)$ denote $\min\limits_{|S| = r} |B(S)|$.
Then for all $1 \leq r \leq \mbinom{n}{k}$,
\begin{align*}
\mu_{K(n,k)}(r) \geq \dbinom{n}{k} - \dfrac{1}{r}\dbinom{n-1}{k-1}^2 - r
\end{align*}
\end{lemma}

We include a proof for completeness.  In it, we make use of the
following variant of the Erd\H{o}s-Ko-Rado theorem.  Call the set
families $\X$ and $\Y$ {\em
  cross-intersecting} if $X \cap Y \neq \emptyset$ for all $X \in \X$
and $Y \in \Y$.

\begin{lemma}[\cite{Matsumoto89:Exact}]\label{lem:cross-intersect}
If $\X$ and $\Y$ are cross-intersecting families of size-$k$ subsets of
$[n]$, then
\begin{align*}
|\X||\Y| \leq \binom{n-1}{k-1}^2
\end{align*}
\end{lemma}
Note that the inequality in Lemma~\ref{lem:cross-intersect} holds with
equality (for $k \le n/2$) when $\X$ and $\Y$ both consist of all
subsets of size $k$ that contain the element~1.

\begin{proof} 
(Of Lemma~\ref{lem:boundary}.)  For any $S$, we can partition $V$ into $S$, $B(S)$, and $V \backslash(S\cup B(S))$. An example of this is shown in Figure~\ref{fig:boundary}. Consider an arbitrary $a \in V\backslash(S\cup B(S))$. We know that $a\not\in S$ and $a\not\in B(S)$, so there is no $b\in S$ where $(a,b) \in E$. Therefore for all $a \in V\backslash(S\cup B(S))$ and $b\in S$, $a$ and $b$ are not adjacent. Recall that $a$ and $b$ are adjacent in $K(n,k)$ if $a\cap b = \emptyset$. Therefore for all $a\in V\backslash(S\cup B(S))$ and $b\in S$, $a\cap b \neq\emptyset$. Thus $S$ and $V\backslash(S\cup B(S))$ are cross-intersecting families.

\colorlet{color1}{red!20}
\colorlet{color2}{blue!20}
\colorlet{color3}{gray!40}
\begin{figure}[ht!bp]
\centering
\begin{tikzpicture}[every node/.style={draw,circle,fill=blue!20,minimum size=6 mm,node distance=1.3cm}]
  \node[fill=color2] (1) {$$};
  \node[fill=color2] (2) [right of=1] {$$};
  \node[fill=color3] (3) [right of=2] {$$};
  \node[fill=color1] (4) [below of=1] {};
  \node[fill=color2] (5) [right of=4] {};
  \node[fill=color3] (6) [right of=5] {};
  \node[fill=color1] (7) [left of=1] {};
  \node[fill=color1] (8) [left of=4] {};

  \path[every node/.style={font=\sffamily\normalsize}]
    (1) edge [left] node [left] {} (7)
    (2)  edge [left] node[left] {} (6)
    (3) edge [left] node[below] {} (6) 
         edge [right] node[above] {} (2)
    (4) edge [left] node[left] {} (1)
    	 edge [left] node[left] {} (2)
	 edge [left] node[left] {} (5)
	 edge [left] node [left] {} (8)
    (5) edge [left] node[left] {} (6)

      ;
\end{tikzpicture}
\caption{The partitioning of an arbitrary graph into $S$, $B(S)$, and $V\backslash(S\cup B(S))$. In this example, $S$ is the set of pink vertices, $B(S)$ is the set of blue vertices, and $V\backslash(S\cup B(S))$ is the set of gray vertices.}\label{fig:boundary}
\end{figure}
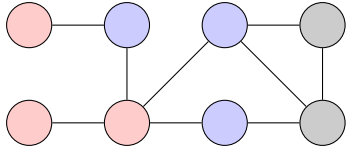

Therefore by Lemma~\ref{lem:cross-intersect}, we have $|S| |V\backslash(S\cup B(S))| \leq \mbinom{n-1}{k-1}^2$. Let $r = |S|$. Then $|V\backslash(S\cup B(S))| \leq \mfrac{1}{r}\mbinom{n-1}{k-1}^2$. Therefore for all $S$,
\begin{align*}
|B(S)| =&\ |V| - |V\backslash(S\cup B(S))| - |S|\\
=&\ \mbinom{n}{k} - |V\backslash(S\cup B(S))| - r\\
\geq&\ \mbinom{n}{k} - \mfrac{1}{r}\mbinom{n-1}{k-1}^2 - r
\end{align*}
and so $\mu_{K(n,k)}(r) = \min\limits_{|S| = r} |B(S)| \geq \mbinom{n}{k} - \mfrac{1}{r}\mbinom{n-1}{k-1}^2 - r$.
\end{proof}

We will only be interested in $K(2k+1, k)$, so we will simply write 
\[
\mu(r) = \mu_{K(2k+1, k)}(r).
\] 
Similarly, let 
\[
\beta(r) = \mbinom{2k+1}{k} - \mfrac{1}{r}\mbinom{2k}{k-1}^2 - r.
\] 
Then $\mu(r) \geq \beta(r)$ for all $r$.

We next prove a lemma building on Lemma~\ref{lem:boundary}.

\begin{lemma}\label{lem:bound-with-min} Let $r_{max} = \mbinom{2k}{k-1}$. Then for the graph $K(2k+1, k)$ and any $r^* \leq r_{max}$,
\begin{align*}
\min_{S} \{|B(S)|: r^* \leq |S| \leq r_{max}\} \geq \beta(r^*)
\end{align*}
\end{lemma}

\begin{proof}
We begin by examining the expression $\beta(r) - \beta(r-1)$:
\begin{align*}
\beta(r) - \beta(r-1) =&\ - \mfrac{1}{r}\mbinom{2k}{k-1}^2 - r  + \mfrac{1}{r-1}\mbinom{2k}{k-1}^2 + r - 1\\
=&\ \Big(\mfrac{1}{r-1} - \mfrac{1}{r}\Big) \mbinom{2k}{k-1}^2 - 1\\
=&\ \mfrac{1}{r(r-1)}\mbinom{2k}{k-1}^2 - 1.
\end{align*}
Therefore $\beta(r) - \beta(r-1) \geq 0$ when $r(r-1) \leq \mbinom{2k}{k-1}^2$. If $r \leq r_{max}$, then $r(r-1) < r^2 \leq r_{max}^2 = \mbinom{2k}{k-1}^2$. Thus $\beta(r) \geq \beta(r-1)$ when $r \leq r_{max}$. Iterating this inequality yields $\beta(r^*) \leq \beta(r)$ whenever $r^* \leq r \leq r_{max}$.

We can rewrite $\min\limits_{S} \{|B(S)|: r^* \leq |S| \leq r_{max}\}$ as
\begin{align*}
\min\limits_{S} \{|B(S)|: r^* \leq |S| \leq r_{max}\} =&\
\min\limits_{r:\ r^* \leq r \leq r_{max}}\ \min\limits_{|S| = r}\ |B(S)|\\
=&\ \min\limits_{r:\ r^* \leq r \leq r_{max}} \mu(r) \\
\geq&\ \min\limits_{r:\ r^* \leq r \leq r_{max}} \beta(r)
\end{align*}
where the last step is due to Lemma~\ref{lem:boundary}. Since $\beta(r^*) \leq \beta(r)$ whenever $r^* \leq r \leq r_{max}$, $\min\limits_{r:\ r^* \leq r \leq r_{max}} \beta(r) = \beta(r^*)$. Therefore $\min\limits_{S} \{|B(S)|: r^* \leq |S| \leq r_{max}\} \geq \beta(r^*)$, as required.
\end{proof}

\subsubsection{Local search on $K(2k + 1, k)$}
We are now ready to prove our result on $D[LS(K(2k + 1, k))]$.

\begin{theorem}\label{thm:local-search-complexity} For all $k$,
\begin{align*}
D[LS(K(2k+1, k))] \in \Omega\left(\dfrac{1}{k} \dbinom{2k+1}{k}\right).
\end{align*}
\end{theorem}

\begin{proof}
Let $c = r_{max} = \mbinom{2k}{k-1}$, and let $t =
\mfrac{1}{2k+1}r_{max}$. so $c-t = \mfrac{2k}{2k+1}r_{max}$. Then by
Lemma~\ref{lem:sep-game},
\[
D[LS(K(2k+1, k))] \geq\]\[
\min \Big( \mfrac{1}{2k+1}r_{max}, \min\limits_S
\big\{|B(S)|: \mfrac{2k}{2k+1}r_{max} \leq |S| \leq
r_{max}\big\} \Big)
\]

By Lemma~\ref{lem:bound-with-min},
\begin{align*}
\ \min\limits_S \big\{|B(S)|:&\ \mfrac{2k}{2k+1}r_{max} \leq |S| \leq r_{max}\big\} 
\geq \beta\left(\mfrac{2k}{2k+1} r_{max}\right)\\
=&  \mbinom{2k+1}{k} - \mfrac{2k+1}{2k\cdot r_{max}} r_{max}^2 - \mfrac{2k}{2k+1}r_{max}\\
\geq&\ \mbinom{2k+1}{k} - \left(\mfrac{2k+1}{2k} + 1\right)r_{max}.
\end{align*}

Using the identity $\mbinom{n}{k}= \mfrac{n}{k}\mbinom{n-1}{k-1}$ for any $n,k$, we have $\mbinom{2k+1}{k} = \mfrac{2k+1}{k} \mbinom{2k}{k-1} = \mfrac{2k+1}{k}r_{max}$. Thus we have
\begin{align*}
\min\limits_S \big\{|B(S)|: \mfrac{2k}{2k+1}r_{max}&\ \leq |S| \leq r_{max}\big\}\\
\geq&\ \left(\mfrac{2k+1}{k} - \mfrac{2k+1}{2k} - 1\right)r_{max}\\
=& \mfrac{4k+2 - 2k-1 - 2k}{2k} r_{max}\\
=&\ \mfrac{1}{2k} r_{max}.
\end{align*}
Therefore,
\begin{align*}
 D[LS(K(2k+1, k))] \geq&\ \min\left(\mfrac{1}{2k+1}r_{max}, \mfrac{1}{2k}r_{max}\right)\\
 = &\ \mfrac{1}{2k+1}r_{max}\\
  \in&\ \Omega\left(\mfrac{1}{k} r_{max}\right).
 \end{align*}
Since $\mbinom{2k+1}{k} = \mfrac{2k+1}{k}r_{max}$, we have
\[
D[LS(K(2k+1, k))] \in \Omega\left(\mfrac{1}{k} \mbinom{2k+1}{k}\right).
\]
\end{proof}

Theorem~\ref{thm:efx-reduction} and Theorem~\ref{thm:local-search-complexity} together imply our main result of Section~\ref{sec:lower-bound}.
\begin{theorem}\label{thm:efx-complexity} 
The deterministic query complexity of the \textsc{EFX Allocation}
problem satisfies
\begin{align*}
D[EFX_{id}(2,2k+1)] \in \Omega\left(\dfrac{1}{k} \dbinom{2k+1}{k}\right),
\end{align*}
even for two players with identical submodular valuations.
\end{theorem}

\subsection{Randomized query complexity}\label{sec:rand-complexity}

Our reduction from \textsc{Local Search} to \textsc{EFX Allocation} also yields an exponential lower bound for randomized algorithms
for free, thanks to results due to Dinh and Russell~\shortcite{Dinh10:Quantum} and Valencia-Pabon and Vera~\shortcite{Valencia-Pabon05:Diameter}. Let $R[LS(G)]$ be the minimum
number of queries required to solve \textsc{Local Search} on $G$ by a
randomized algorithm: the algorithm should output a local
  maximum with probability at least $2/3$ (say) over its internal coin
  flips. Formally, $R[LS(G)] = \min\limits_{\Gamma_R} \max\limits_f T(G, f, \Gamma_R)$, where $\Gamma_R$ ranges over the set of randomized algorithms.

Similarly, let $R[EFX_{id}(2,2k+1)]$ be the minimum number of queries required by
a randomized algorithm to find an EFX allocation for two players with identical valuations, and $2k+1$ goods (again with correctness probability at least
2/3, say).

\begin{theorem}[\cite{Dinh10:Quantum}]\label{thm:general-rand} If $G = (V,E)$ is a vertex transitive graph with diameter $d$, then
\begin{align*}
R[LS(G)] \in \Omega\Big(\frac{\sqrt{|V|}}{d \cdot \log |V|} \Big)
\end{align*}
\end{theorem}

Since $K(2k + 1, k)$ is vertex transitive, the last piece of the puzzle is the following theorem,

\begin{theorem}[\cite{Valencia-Pabon05:Diameter}]\label{thm:diameter}
The diameter of $K(2k+1, k)$ is $k$.
\end{theorem}

With these two tools in hand, Theorem~\ref{thm:efx-rand} requires only a short proof. 

\begin{theorem}\label{thm:efx-rand} The randomized query complexity of the \textsc{EFX Allocation} problem satisfies
\begin{align*}
R[EFX_{id}(2, 2k+1)] \in \Omega\Big(\sqrt{\mbinom{2k+1}{k}} \mfrac{1}{k^2} \Big)
\end{align*}
even for two players with identical submodular valuations.
\end{theorem}

\begin{proof}
Since $|V| = \mbinom{2k+1}{k}$ and $\log \left(\mbinom{2k+1}{k}\right) \in O(\log (4^k)) = O(k)$, we have
\[
R[LS(K(2k+1, k))] \in \Omega\Big(\sqrt{\mbinom{2k+1}{k}} \mfrac{1}{dk} \Big)
\]
by Theorem~\ref{thm:general-rand}. Thus by Theorem~\ref{thm:diameter}, we have
\[
R[LS(K(2k+1, k))] \in \Omega\Big(\sqrt{\mbinom{2k+1}{k}} \mfrac{1}{k^2} \Big)
\]

The reduction used to prove that $D[EFX_{id}(2, 2k+1)] \geq D[LS(K(2k+1,k))]$
can equivalently be used to show that
\[
R[EFX_{id}(2, 2k+1)] \geq R[LS(K(2k+1,k))].
\] 
Therefore $R[EFX_{id}(2, 2k+1)] \in \Omega\Big(\sqrt{\mbinom{2k+1}{k}} \mfrac{1}{k^2} \Big)$.
\end{proof}

While this bound is not as strong as our deterministic lower bound (Theorem~\ref{thm:efx-complexity}), it does establish that even a randomized algorithm requires an exponential number of queries to find an EFX allocation.

%% file: genIdUpperBounds.tex
\section{Existence of EFX allocations for general but identical valuations}\label{sec:gen-id}

We mentioned in the previous section that an EFX allocation is guaranteed to exist when all players have the same valuation: this section proves that claim. Specifically, we show that a modified version of the leximin solution is guaranteed to be EFX for general but identical valutions. This also yields a cut-and-choose-based protocol for two players with general and possibly distinct valuations.

\subsection{The leximin solution}
The leximin solution selects the allocation which maximizes the minimum utility of any player. If there are multiple allocations which achieve that minimum utility, it chooses among those the one which maximizes the second minimum utility, and so on. This implicitly specifies a comparison operator $\prec$, which is given by Algorithm~\ref{alg:leximin}, and constitutes a total ordering over allocations. 

The operator $\prec$ takes as input two allocations $A$ and $B$, and the list of player valuations $(v_1...v_n)$. The players are ordered by utility, and according to some arbitrary but consistent tiebreak for players with the same utility (for example, by player index). The comparison terminates when the $\ell$th player in $A$'s ordering $X^A$ has different utility from the $\ell$th player in $B'$s ordering $X^B$.

The leximin solution is the global maximum under this ordering. The leximin solution is trivially PO, since if it were possible to improve the utility of one player without decreasing the utility of any other player, the new allocation would be strictly larger under $\prec$.

\subsubsection{Standard leximin is not EFX}
Unfortunately, the standard leximin solution is not always EFX, even for identical valuations. Consider two players with the same (non-additive) valuation $v$ over two goods $a$ and $b$. Define $v$ by
\[ v(S) =
   \begin{cases} 
      0 & \textnormal{if}\ \ S = \{a\}\\
      1 & \textnormal{if}\ \ S = \{b\}\\
      2 & \textnormal{if}\ \ S = \{a,b\}\\
   \end{cases}
\]

By symmetry, suppose without loss of generality that player 1 receives good $b$. Define the allocation $A$ by $A_1 = \{b\}$ and $A_2 = \{a\}$, and define the allocation $B$ by $B_1 = \{a,b\}$ and $B_2 = \emptyset$.

Since player 2 (the minimum utility player) is indifferent between $A$ and $B$, leximin selects allocation $B$ because it maximizes the value of player 1 (the second minimum utility player). However, $A$ is EFX, while $B$ is not: player 2 envies player 1 even after the removal of $a$ from $B_1$.\footnote{This example will be relevant again in Section~\ref{sec:po} as an instance where there is no allocation which is both EFX and PO.} 

\subsection{The leximin++ solution}
Our fix is that after maximizing the minimum utility, we maximize the size of the bundle of the player with minimum utility, before maximizing the second minimum utility. Then we maximize the second minimum utility, followed by the size of the second minimum utility bundle, and so on. Thus giving good $a$ to the lower utility player (player 2) is preferable, and so the EFX allocation $A$ is chosen over $B$.

We call this the \emph{leximin++ solution}. The leximin++ solution induces a comparison operator $\prec_{++}$, also given in Algorithm~\ref{alg:leximin}. Similarly to $\prec$, the players are ordered by increasing utility, and then according to an arbitrary but consistent tiebreak among players with the same utility.\footnote{The tiebreak method must be consistent to ensure that $\prec_{++}$ is a total ordering. Consider two players with the same valuation $v$, and a single good $a$ where $v(\{a\}) = 0$. Suppose $a \in A_1$. Since both players have zero utility, if the tiebreak method were not required to be consistent, both $\{1,2\}$ and $\{2,1\}$ would be valid player orderings for $A$. Consider running $A \prec_{++} A$. If player 2 were considered first in the $A$ on the left, and player 1 were considered first in the $A$ on the right, the operator would see that player 1 has a larger bundle than player 2, and return true.}
The comparison terminates when the $\ell$th player in $X^A$ differs in utility or bundle size from the $\ell$th player in $X^B$, with utility being checked before bundle size.

\begin{algorithm*}[tb]
\centering
\begin{algorithmic}[1]
  \Function{LeximinCmp}{$A, B, (v_1...v_n)$} \Comment{Returns true if $A \prec B$ (strictly)}
  \State $X^A \gets$ ordering of players by increasing utility $v_i(A_i)$, then by some arbitrary but consistent tiebreak method for players with the same utility
  \State $X^B \gets$ corresponding ordering of players under $B$
  \ForAll {$\ell \in [n]$}
  	\State $i \gets X^A_\ell$ \Comment{$\ell$th player in the ordering $X_A$}
	\State $j \gets X^B_\ell$ \Comment{$\ell$th player in the ordering $X_B$}
	\If {$v_i(A_i) \neq v_j(B_j)$}
		\Return $v_i(A_i) < v_j(B_j)$
	\EndIf
  \EndFor
  \Return false \Comment{In this case, $A$ and $B$ are equal}
\EndFunction
\end{algorithmic}
  
\begin{algorithmic}[2]
  \Function{Leximin++Cmp}{$A, B, (v_1...v_n)$} \Comment{Returns true if $A \prec_{++} B$ (strictly)}
  \State $X^A \gets$ same as in \textsc{LeximinCmp}
  \State $X^B \gets$ same as in \textsc{LeximinCmp}
  \ForAll {$\ell \in [n]$}
  	\State $i \gets X^A_\ell$
	\State $j \gets X^B_\ell$
	\If {$v_i(A_i) \neq v_j(B_j)$}
		\Return $v_i(A_i) < v_j(B_j)$
	\EndIf
	\If {$|A_i| \neq |B_j|$}
		\Return $|A_i| < |B_j|$
	\EndIf
  \EndFor
  \Return false
\EndFunction
\end{algorithmic}
\caption{Leximin and Leximin++ comparison operators}
\label{alg:leximin}
\end{algorithm*}

It may not be immediately clear that $\prec_{++}$ specifies a total ordering, but this is in fact the case. The proof of Theorem~\ref{thm:total} appears in Appendix Section~\ref{sec:add-proofs}.

\begin{theorem}\label{thm:total}
The comparison operator $\prec_{++}$ specifies a total ordering.
\end{theorem}

We are now ready to prove our main result of this section.

\begin{theorem}\label{thm:general-identical}
For general but identical valuations, the leximin++ solution is EFX.
\end{theorem}

\begin{proof}
Let $A$ be an allocation that is not EFX. We will show that $A$ is not the leximin++ solution.

Since $A$ is not EFX, there exist players $i,j$ and $g\in A_j$ where $v(A_i) < v(A_j \backslash\{g\})$. Then any player with utility $\min_k v(A_k)$ must also have utility strictly less than $v(A_j \backslash\{g\})$, so assume with loss of generality that $i = \argmin_k v(A_k)$. If there are multiple players with minimum utility in $A$, let $i$ be the one considered last in the ordering $X^A$, according to the arbitrary but consistent tiebreak method.

Define a new allocation $B$ where $B_{i} = A_{i} \cup \{g\}$, $B_j = A_j \backslash \{g\}$, and $B_k = A_k$ for all $k \not\in \{i,j\}$. We will show that $A \prec_{++} B$. 

Let $S$ be the set of players appearing before $i$ in $X^A$. We know $i$ is considered last among the players with minimum utility by assumption, so $S$ is exactly the set of players with minimum utility, other than $i$. Note that neither $i$ nor $j$ are in $S$.

Since the only bundles that differ between allocations $A$ and $B$ are that of $i$ and $j$, we have $A_k = B_k$ for all $k\in S$. Thus for all $k\in S$, $v(B_k) = v(A_k) = v(A_i)$. Since $v(B_j) > v(A_i)$, $j$ must occur after every player in $S$ in $X^B$. 

Because $A_i \subset B_i$, we have $v(B_i) \geq v(A_i)$. If $v(B_i) > v(A_i)$, $i$ must occur after every player in $S$ in $X^B$, since $v(B_i) > v(B_k)$ for all $k\in S$. If $v(B_i) = v(A_i)$, $i$ is still considered after every player in $S$ according to the arbitrary but consistent tiebreak method. Thus $i$ occurs after every player in $S$ in $X^B$ in either case, which shows that the first $|S|$ players in $X^B$ are the players in $S$, in the same order they occur in $X^A$.

Therefore the leximin++ comparison will not have terminated before reaching position $|S| + 1$ in the orderings. Let $T$ be the set of players appearing after $i$ in $X^A$: note that $j\in T$. By assumption, of the players with minimum utility in $A$, $i$ appears last in $X^A$. Therefore all players after $i$ in $X^A$ do not have minimum utility, so $v(A_k) > v(A_i)$ for all $k \in T$. Recall that $v(B_j) > v(A_i)$ and that for all $k\in T\backslash\{j\}$, $v(B_k) = v(A_k)$. Thus $v(B_k) > v(A_i)$ for all $k\in T$.

We know that $X^A_{|S| + 1} = i$. If $X^B_{|S| + 1} = i$, we have $|A_i| < |B_i|$ (and possibly also $v(A_i) < v(B_i)$), so $A \prec_{++} B$ returns true. If $X^B_{|S| + 1} = k$ for some $k \neq i$, then $k \in T$. Therefore $v(A_i) < v(B_k)$, so $A \prec_{++} B$ returns true in this case as well.

Since $A \prec_{++} B$, $A$ cannot be the leximin++ solution. Therefore the leximin++ solution must be EFX.
\end{proof}
We now show how Theorem~\ref{thm:general-identical} can easily be used to find an EFX allocation for two players with general and
possibly distinct valuations.\footnote{The two-player case is not
  trivial.  For example, our lower bound in
  Theorem~\ref{thm:efx-complexity} already applies with two players
  (even with identical valuations).}
Our algorithm for this follows from the
observation that any player can partition the goods into $k$ bundles
that are mutually EFX from her viewpoint, simply by computing the leximin++ solution with $k$ copies of herself.

Algorithm~\ref{alg:two-players} is a straightforward adaptation of the
cut-and-choose protocol. Player $1$ partitions the goods into two
bundles using the leximin++ solution, and player $2$
chooses her favorite bundle.

\begin{algorithm*}[tb]
\centering
\begin{algorithmic}[1]
  \Function{CutAndChoose}{$m, v_1, v_2$}
  \State $(A_1, A_2) \gets \text{Leximin++Solution}(2,m,v_1)$ 
  	\Comment{Player 1 uses the leximin++ solution to cut,}
  \If{$v_2(A_1) \geq v_2(A_2)$} \Comment{and player 2 chooses.}
  	\Return $(A_2, A_1)$
  \Else 
 	 \Return $(A_1, A_2)$
  \EndIf
\EndFunction
\end{algorithmic}
\caption{Find an EFX allocation for two players with general valuations via cut-and-choose}
\label{alg:two-players}
\end{algorithm*}

\begin{theorem}\label{thm:two-players}
  For two players with general (not necessarily identical) valuations,
  Algorithm~\ref{alg:two-players} returns an EFX allocation.
\end{theorem}

\begin{proof}
By Theorem~\ref{thm:general-identical}, the allocation is EFX from player $1$'s viewpoint regardless of which bundle she receives. Player $2$ receives her favorite bundle, so the resulting allocation is EFX from her viewpoint as well.
\end{proof}

\subsection{Limitations of leximin++}
Unfortunately, the leximin++ solution may not be EFX when players have different valuations. For example, consider two players with valuations $v_1(S) = |S|$ and $v_2(S) = \epsilon |S|$, for some small $\epsilon > 0$. As long as player 1 receives at least one good, she will have utility at least 1. However, player 2 will always have utility less than 1 for a suitably small $\epsilon$. Thus the leximin++ solution gives a single good to player 1 and the rest to player 2, which will cause player 1 to envy player 2 in violation of EFX.

One might hope that this could be remedied by assuming that all players have the same value for the entire set of goods (or rescaling valuations as necessary if this is not the case). Unfortunately, the set of additive valuations given by Figure~\ref{fig:leximin-fails} thwarts this hope. \begin{figure}[h]
\centering
\begin{tabular}{c|cccc} 
& a & b & c & d\\
\hline
player 1 & 14 & 3 & 2 & 1 \\
player 2 & 7 & 6 & 4 & 3\\
player 3 & 20 & 0 & 0 & 0
\end{tabular}
\caption{An example where the leximin++ solution fails to be EFX even when all players have the same value for the entire set of goods.}
\label{fig:leximin-fails}
\end{figure}

We claim that the allocation $A = (\{b,d\}, \{c\}, \{a\})$ is the only allocation where all players have utility at least $4$. To see this, first observe that good $a$ must go to player 3, or player 3 has zero utility. Then the only way to give players 1 and 2 each utility at least $4$ is to give $\{b,d\}$ to player 1 and $\{c\}$ to player 2.

Since $A$ is the only allocation which gives all players utility at least $4$, $A$ must be the leximin++ solution. However, $A$ is not EFX, because $v_2(\{c\}) < v_2(\{b,d\} \backslash\{d\})$.

We mentioned at the beginning of this section that the leximin solution is trivially PO. The leximin++ solution does not share this guarantee. Indeed, this is necessary in order for the leximin++ solution to be EFX, since it is impossible to simultaneously guarantee EFX and Pareto optimality, even for identical valuations (Theorem~\ref{thm:no-efx-po-gen}). However, that example relies on zero value goods. We will show in the next section that if zero value goods are disallowed, the leximin solution becomes EFX as well as PO in two contexts.

%% file: PO.tex
\section{Pareto optimality}\label{sec:po}

In this section, we examine when EFX and Pareto optimality can be guaranteed simultaneously. We begin by showing that if a player is wholly indifferent to a good being added to her bundle (zero marginal utility), EFX and Pareto optimality can be mutually exclusive even in simple cases.

\begin{theorem}\label{thm:no-efx-po-add}
If zero marginal utility is allowed, there exist additive valuations where no EFX allocation is also PO, even for two players.
\end{theorem}

\begin{proof}
Consider the following additive valuations:
\begin{center}
\begin{tabular}{ c|ccc} 
& a & b & c\\
\hline
player 1 & 2 & 1 & 0 \\
player 2 & 2 & 0 & 1\\
\end{tabular}
\end{center}

Since $v_1(\{c\}) = 0$ but $v_2(\{c\}) > 0$, $c \in A_2$ in any PO allocation. Similarly, $b \in A_1$ in any PO allocation.

By symmetry, assume without loss of generality that $a \in A_1$, so $A_1 = \{a,b\}$ and $A_2 = \{c\}$. Then $v_2(\{c\}) = 1$, but $v_2(A_1 \backslash\{b\}) = v_2(\{a\}) = 2$, so the allocation is not EFX.

Therefore no allocation is both EFX and PO.
\end{proof}

A similar example exists for general and identical valuations. This example was also used in Section~\ref{sec:gen-id} to show that the leximin solution may not be EFX when zero marginal utility is allowed.

\begin{theorem}\label{thm:no-efx-po-gen}
If zero marginal utility is allowed, there exist general and identical valuations where no EFX allocation is also PO, even for two players.
\end{theorem}

\begin{proof}
Consider two players with the same valuation $v$, and two goods $a$ and $b$. Define $v$ by
\[ v(S) =
   \begin{cases} 
      0 & \textnormal{if}\ \ S = \{a\}\\
      1 & \textnormal{if}\ \ S = \{b\}\\
      2 & \textnormal{if}\ \ S = \{a,b\}\\
   \end{cases}
\]
By symmetry, assume without loss of generality that $b \in A_1$. If $A_1 = \{a,b\}$, then $v(A_2) = v(\emptyset) = 0$, but $v(A_1 \backslash \{a\}) = v(\{b\}) > 0$, so the allocation is not EFX.

Therefore in any EFX allocation, $a \in A_2$. But $v(\{a\}) = v(\emptyset) = 0$ and $v(\{a,b\}) > v(\{b\})$. Thus giving $a$ to player 1 strictly increases player 1's value, without changing player 2's value, so the allocation is not PO.

Therefore no allocation is both EFX and PO.
\end{proof}

On the other hand, if valuations are required to be additive and identical, it is possible to guarantee EFX and Pareto optimality simultaneously, even with zero marginal utility. However, this is an extremely restrictive setting that we mention mostly for completeness; we consider this a very minor result. The proof of Theorem~\ref{thm:yes-efx-po-add-id} appears in Section~\ref{sec:add-proofs} of the appendix.

\begin{theorem}\label{thm:yes-efx-po-add-id}
For additive and identical valuations, there exists an allocation that is both EFX and PO (even allowing zero marginal utility).
\end{theorem}

\subsection{Nonzero marginal utility}

The negative results of Theorem~\ref{thm:no-efx-po-add} and Theorem~\ref{thm:no-efx-po-gen} both break down if players are assumed to have strictly positive utility for any good being added to their bundle. Formally, we say that a valuation $v$ has nonzero marginal utility if for every set $S \subset [m]$ and $g \not\in S$, $v(S \cup \{g\}) - v(S) > 0$.

We feel that this is a reasonable assumption in practice, as $v(S \cup \{g\}) - v(S)$ is allowed to be arbitrarily small, and one might expect players in real world situations to always prefer to have a good than not.
\subsubsection{Positive results from leximin}\label{sec:leximin-po-gen-id} Under the assumption of nonzero marginal utility, the leximin solution is guaranteed to be both EFX and PO for any number of players with general but identical valuations, and for two players with (possibly distinct) additive valuations.

\begin{theorem}\label{thm:gen-id-po}
For general but identical valuations with nonzero marginal utility, the leximin solution is EFX and PO.
\end{theorem}

\begin{proof}
We follow a very similar analysis to the proof of Theorem~\ref{thm:general-identical}. Let $A$ be an allocation that is not EFX. Then there exist players $i,j$ and $g\in A_j$ where $v(A_i) < v(A_j \backslash\{g\})$. Again assume without loss of generality that $i = \argmin_k v(A_k)$, and if there are multiple players with minimum utility in $A$, let $i$ be the one considered last in the ordering $X^A$.

Define the same new allocation $B$ where $B_{i} = A_{i} \cup \{g\}$, $B_j = A_j \backslash \{g\}$, and $B_k = A_k$ for all $k \not\in \{i,j\}$. When zero marginal utility is allowed, the leximin++ modification of considering bundle size is necessary because otherwise if $v_i(B_i) = v_i(A_i)$, it could be the case that $B \prec A$. When zero marginal utility is disallowed, this modification is not necessary because $v_i(B_i) > v_i(A_i)$ always.

The proof of Theorem~\ref{thm:general-identical} can be used nearly verbatim to show that $A \prec B$ (simply omit the sentences handling the case where $v(B_i) = v(A_i)$, since we now have $v(B_i) > v(A_i)$, due to the nonzero marginal utility of $v$). Thus $A$ is not the leximin solution, so the leximin solution is EFX.

As noted before, the leximin solution is trivially Pareto optimal, since if any player could be made better off without hurting any other player, that new allocation would be strictly larger under $\prec$.
\end{proof}

We now show that assuming nonzero marginal utility, the leximin solution is EFX and PO for two players with additive valuations. For this theorem, we will assume that $v_i([m]) = 1$ for all $i$: were this not the case, we could easily define $v'_i(S) = v_i(S)/v_i([m])$, and find the leximin solution according to $v'$. Additivity is necessary for Theorem~\ref{thm:two-players-po} so that $v_i(A_1) < v_i(A_2)$ implies $v_i(A_1) < 1/2$, and so that  $v_i(A_1) \geq v_i(A_2)$ implies $v_i(A_1) \geq 1/2$.

The proof is similar to those of Theorem~\ref{thm:general-identical} and Theorem~\ref{thm:gen-id-po}, in that we consider an arbitrary allocation $A$ that is not EFX, and show that it cannot be the leximin solution by constructing an allocation $B$ such that $A \prec B$. However, the allocation $B$ is constructed differently here.

\begin{theorem}\label{thm:two-players-po}
For two players with additive valuations (not necessarily identical) with nonzero marginal utility, the leximin solution is EFX and PO.
\end{theorem}

\begin{proof}
Let $A$ be an allocation that is not EFX. Then there exist players $i,j$ and $g\in A_j$ where $v_i(A_i) < v_i(A_j \backslash\{g\})$. Without loss of generality, assume $i=1$ and $j=2$.

We know that $v_1(A_1) < v_1(A_2)$, so $v_1(A_1) < 1/2$. If $v_2(A_2) < v_2(A_1)$, the players could swap bundles to increase both of their utilities, so $A$ could not be the leximin solution. Therefore assume $v_2(A_2) \geq v_2(A_1)$, and so $v_2(A_2) \geq 1/2$.

Define two new bundles $S_1 = A_1 \cup \{g\}$ and $S_2 = A_2\backslash\{g\}$. Then define a new allocation $B$ where $B_1 = \argmin\limits_{S \in \{S_1, S_2\}} v_2(S)$ and $B_2 = \argmax\limits_{S \in \{S_1, S_2\}} v_2(S)$.

Since player 2 received her favorite of $S_1$ and $S_2$, we still have $v_2(B_2) \geq 1/2$. We have $v_1(S_2) = v_1(A_2 \backslash \{g\}) > v_1(A_1)$ by our original assumption that $A$ is not EFX, and we have $v_1(S_1) = v_1(A_1 \cup \{g\}) > v_1(A_1)$ by the nonzero marginal utility of $v_1$. Therefore regardless of which bundle player 1 receives, $v_1(B_1) > v_1(A_1)$. 

Thus $B$ has a higher minimum utility than $A$, so $A$ cannot be the leximin solution. Therefore the leximin solution is EFX in this setting, and it remains trivially PO.
\end{proof}

Assuming nonzero marginal utility, Theorem~\ref{thm:two-players-po} provides stronger guarantees than the currently deployed algorithm on Spliddit, which only guarantees an EF1 and PO allocation. As described in Section~\ref{sec:outperforming}, this manifests even in simple cases. 

We also argue that the assumption of nonzero marginal utility is particularly reasonable in the case of two players with additive valuations, since if a player is truly indifferent to some good, perhaps that good could simply be given to the other player and excluded from the fair division process entirely.


\subsubsection{Counterexample for two players with general valuations} Finally, we show that EFX and Pareto optimality cannot be guaranteed simultaneously for general and distinct valuations, even with the assumption of nonzero marginal utility.

\begin{theorem}\label{thm:po-counterexample}
There exist general valuations where no EFX allocation is also PO, even for two players with nonzero marginal utility.
\end{theorem}

\begin{proof}
We construct a set of valuations for which there is no EFX allocation that is also PO.

Let $n=2$ and $M=\{a,b,c,d,e\}$. Let $\alpha_1 = \{a\}, \beta_1 = \{b,d\}, \gamma_1 = \{a,c,d\}$ and $\alpha_2 = \{b\}, \beta_2 = \{a,d\}, \gamma_2 = \{b,d,e\}$. The key properties will be $\alpha_1 \subset \beta_2 \subset \gamma_1$ and $\alpha_2 \subset \beta_1 \subset \gamma_2$.

Define each player's valuation $v_i$ by
\[ v_i(S) =
   \begin{cases} 
      3 + \epsilon(|S| - 3) & \textnormal{if}\ \ \gamma_i \subseteq S\\
      2 + \epsilon(|S| - 2) & \textnormal{if}\ \ \beta_i \subseteq S\ 
      		\textnormal{and}\ \gamma_i \not\subseteq S\\
      1 + \epsilon(|S| - 1) & \textnormal{if}\ \ \alpha_i \subseteq S
      		\ \textnormal{and}\ \beta_i, \gamma_i \not\subseteq S\\
	\epsilon|S| & \textnormal{otherwise}
   \end{cases}
\]
where $\epsilon$ is some small positive value (.1 would suffice). Adding a good to a bundle always increases the value of the bundle by at least $\epsilon$, so $v_i$ satisfies nonzero marginal utility. Also, note that the valuations are symmetric across players, since $\alpha_i, \beta_i$, and $\gamma_i$ are symmetric across players.

We have the following implications:
\begin{align*}
\gamma_i \not\subseteq S \implies&\ v_i(S) < 3\\
\beta_i,\gamma_i \not\subseteq S \implies&\ v_i(S) < 2\\
\alpha_i,\beta_i,\gamma_i \not\subseteq S\implies&\ v_i(S) < 1
\end{align*}

By Theorem~\ref{thm:two-players}, an EFX allocation $A = (A_1, A_2)$ must exist. Suppose $\gamma_i \subseteq A_i$ for some $i$: by symmetry, suppose $i=1$. Since $\beta_1 \cap \beta_2 \cap \gamma_1 \cap \gamma_2 = \{d\} \neq \emptyset$, we have $\beta_2, \gamma_2 \not\subseteq A_2$, so $v_2(A_2) < 2$. Furthermore, $\beta_2$ is a strict subset of $A_1$: specifically, $\beta_2\subseteq A_1\backslash \{c\}$. Therefore $v_2(A_1 \backslash \{c\}) \geq v_2(\beta_2) = 2$, which is strictly larger than $v_2(A_2)$. Therefore if $\gamma_i \subseteq A_i$ for either $i$, $A$ is not EFX.

Now suppose $\beta_i \subseteq A_i$ for some $i$: again suppose $i=1$. Similarly, $\beta_2, \gamma_2 \not\subseteq A_2$. In this case, we also have $\alpha_2 \not\subseteq A_2$, since $\alpha_2 \cap \beta_1 \neq \emptyset$. Therefore $v_2(A_2) < 1$. Since $\alpha_2 \subseteq A_1 \backslash \{d\}$, we have $v_2(A_1 \backslash\{d\}) \geq v_2(\alpha_1) = 1$, which is strictly larger than $v_2(A_2)$. Therefore if $\beta_i \subseteq A_i$ for either $i$, $A$ is not EFX. Since $A$ is EFX by assumption, we have $\beta_i, \gamma_i \not\subseteq A_i$ for both $i$, and so $v_i(A_i) < 2$ for both $i$. 

We next claim that $\alpha_i \subseteq A_i$ for both $i$. Suppose $\alpha_1 \not\subseteq A_1$: then $\alpha_1 \subseteq A_2$. Therefore $v_1(A_1) < 1$, and $v_1(A_2) \geq 1$, so player 1 envies player 2. If there exists $g \in A_2\backslash\alpha_1$, then $g$ could be removed and player 1 would still envy player 2. Thus if $|A_2| \geq 2$, $A$ is not EFX, so we have $|A_2| = 1$. But then $\alpha_2 \subseteq A_1$ and $|A_1| \geq 2$, so player 1 is envied in violation of EFX. Thus we have $\alpha_1 \subseteq A_1$, and by symmetry, $\alpha_2 \subseteq A_2$.

One of the players has at least three goods; by symmetry, suppose $|A_1| \geq 3$. Since $\alpha_1 \subseteq A_1$ and $\beta_1, \gamma_1, \alpha_2 \not\subseteq A_1$, we have $A_1 = \{a,c,e\}$ and $A_2 = \{b,d\}$. 

Consider the allocation $B = (B_1, B_2) = (\{a,c,d\}, \{b,e\})$. Player 2 is indifferent between $\{b,d\}$ and $\{b,e\}$, so $v_2(B_2) = v_2(A_2)$. But $\gamma_1 \subseteq B_1$, so $v_1(B_1) > v(A_1)$. Thus player 1 is strictly better off in $B$, and no player is worse off. Therefore $A$ is not PO, and so no EFX allocation is PO.
\end{proof}

One last attempt to salvage EFX and PO in this setting might be to require a strict ranking over bundles, i.e., not allow player 2 to be indifferent between $\{b,d\}$ and $\{b,e\}$. However, even that would not work, because we can easily set $v_2(\{b,e\}) > v_2(\{b,d\})$, in which case both players are strictly better off in $B$.

This counterexample and our query complexity lower bound show that EFX is a very demanding fairness property, even for two players. In the next section, we complement these negative results by showing that an approximate version of EFX is satisfiable for any number of players with subadditive valuations.

%% file: apxAlgo.tex
\section{Existence of $\frac{1}{2}$-EFX allocations for subadditive valuations}\label{sec:apx-algo}
The possible existence of EFX allocations for possibly distinct valuations and $n\geq 3$ remains an open question, even for additive valuations. However, we are able to achieve an approximate version of EFX, for any number of players with (possibly distinct) subadditive valuations. Recall that an allocation $A$ is $c$-EFX if for all $i,j$, and for all $g\in A_j$, $v_i(A_i) \geq c \cdot v_i(A_j \backslash \{g\})$.
In words, an allocation is $c$-EFX if for all $i,j$, and $g\in A_j$, $i$'s value for her own bundle is at least $c$ times her value for $j$'s bundle after removing $g$. For example, $1$-EFX is equivalent to standard EFX. In this section, we give an algorithm that is guaranteed to return a $\frac{1}{2}$-EFX allocation for any number of players with subadditive valuations. 

To describe our algorithm, we must first define the \emph{envy graph}. The envy graph of an allocation $A$ has a vertex for each player, and a directed edge from $i$ to $j$ if player $i$ envies player $j$. Here we mean full envy (i.e. $v_i(A_i) < v_i(A_j)$), not just envy in violation of EFX. It will be necessary for the envy graph in our algorithm to be acyclic; we now show that we can always ensure this. The following lemma is adapted from Lipton et al.~\shortcite{Lipton04:Approximately}.
\begin{lemma}\label{lem:remove-cycles}
Let $A = (A_1, A_2...A_n)$ be a $c$-EFX allocation with envy graph $G = (V,E)$, where $G$ contains a cycle. Then there exists another allocation $B = (B_1, B_2...B_n)$  with envy graph $H$ where $B$ is also $c$-EFX, and $H$ has no cycles.
\end{lemma}

\begin{proof}
We first show that there exists another $c$-EFX allocation $A' = (A'_1...A'_n)$ with envy graph $G'$, where  $G'$ has strictly fewer edges than $G$.

Let $c = (1, 2...|c|)$ be a cycle in $G$. Thus $v_i(A_i) < v_i(A_{(i\ \text{mod}\ |c|) + 1})$ for all $i\in c$. Define a new allocation $A'$ where $A_i' = A_{(i\ \text{mod}\ |c|) + 1}$ for all $i$, and let $G' = (V', E')$ be the envy graph for $A'$. It is clear that $A'$ is a permutation of $A$.

Suppose $A'$ is not $c$-EFX: then there exist  $i, j \in N$ and $g\in A'_j$ where $v_i(A') < c\cdot v_i(A'_j \backslash \{g\})$. Since $A'$ is a permutation of $A$, there exists $k\in N$ where $A_k = A'_j$, so $v_i(A'_i) < c \cdot v_i(A_k \backslash \{g\})$. Observe that $v_i(A'_i) > v_i(A_i)$ if $i\in c$, and $v_i(A'_i) = v_i(A_i)$ otherwise. Thus $v_i(A_i)\leq v_i(A'_i) < c \cdot v_i(A_k \backslash \{g\})$, and so $A$ is also not $c$-EFX. Therefore if $A$ is $c$-EFX, then $A'$ is also $c$-EFX.

Note that the number of edges from $V' \backslash c$ into $c$ is unchanged. Also, the number of edges from $c$ into $V'\backslash c$ has decreased or stayed the same, since the utility of every player in $c$ has strictly increased. Furthermore, for each $i\in c$, the number of players in $c$ whom $i$ envies has decreased by at least one. This shows that $G'$ has strictly fewer edges than $G$.

If $G'$ still contains a cycle, we can apply this process again to obtain $G''$, $G'''$, and so on. Since the number of edges strictly decreases each time, we can apply this process at most $|E|$ times before we obtain a envy graph without a cycle.
\end{proof}

Algorithm~\ref{alg:apx} gives pseudocode for our algorithm. Initially all goods are in the pool $P$, and we proceed in rounds until $P$ is empty, maintaining the invariant that the partial allocation at the end of each round is EFX. The function EliminateEnvyCycles uses Lemma~\ref{lem:remove-cycles} to ensure that the graph at the beginning of each round is acyclic. Since the envy graph is acyclic, we can always find an unenvied player $j$, and give an arbitrary good $g^*$ from $P$ to her. 

It is possible that this will cause another player $i$ to envy $j$ in violation of $\frac{1}{2}$-EFX. In this case, we return all of $i$'s current bundle to $P$, and let $i$'s new bundle be just $\{g^*\}$. The key insight is that in order for $i$ to go from not envying $j$ to envying $j$ in violation of $\frac{1}{2}$-EFX, adding $g^*$ to $A_j$ must have caused $v_i(A_j)$ to at least double. We will use that fact, along with the subadditivity of $v_i$, to show that $v_i(\{g^*\})$ must be larger than $i$'s value for her bundle at the beginning of the round. Thus if $i$ envies any player, it remains consistent with $\frac{1}{2}$-EFX. Any envy directed towards $i$ will be fully EFX, since $i$ will only have one good.

On each round, either $P$ decreases in size (in the case where $g^*$ remains with $j$), or the sum of utilities increases (in the case where $g^*$ is instead given to $i$ because $i$ envies $j$ in violation of $\frac{1}{2}$-EFX). Thus we can use a potential function argument to show that Algorithm~\ref{alg:apx} terminates (although it may take a non-polynomial number of rounds).

\begin{algorithm*}[tb]
\centering
\begin{algorithmic}[1]
  \Function{GetApxEFXAllocation}{$n, m, (v_1...v_n)$}
  \State $P \gets [m]$ \Comment{Initially, all goods are in the pool}
  \ForAll{$i \in [n]$}
  \State $A_i \gets \emptyset$
  \EndFor
  \While{$P \neq \emptyset$}
  \State $g^* \gets \text{pop}(P)$ \Comment{Remove an arbitrary good from $P$,}
  \State $j \gets \text{FindUnenviedPlayer}(A_1, A_2...A_n)$ \Comment{and give it to an unenvied player}
  \State $A_j \gets A_j \cup \{g^*\}$
  \If{$\exists i \in [n], g \in A_j$ such that $v_i(A_i) < \frac{1}{2} v_i(A_j \backslash\{g\})}$
  \State $P \gets P \cup A_i$ \Comment{Return $i$'s old allocation to the pool,}
  \State $A_j \gets A_j \backslash \{g^*\}$ \Comment{and give $i$ just $\{g^*\}$}
   \State $A_i \gets \{g^*\}$
  \EndIf
  \State $(A_1, A_2...A_n) \gets \text{EliminateEnvyCycles}(A_1, A_2...A_n)$ \Comment{Ensure the envy graph is acyclic}
\EndWhile
\Return $(A_1, A_2...A_n)$
\EndFunction
\end{algorithmic}
\caption{Find an $\frac{1}{2}$-EFX allocation for $n$ players with subadditive valuations}
\label{alg:apx}
\end{algorithm*}

\begin{theorem}\label{thm:apx-upper}
For subadditive valuations, Algorithm~\ref{alg:apx} returns a $\frac{1}{2}$-EFX allocation.
\end{theorem}

\begin{proof}
We refer to each iteration of the while-loop as a round. We first show that the partial allocation at the end of each round is $\frac{1}{2}$-EFX. Then we will show that the algorithm is guaranteed to terminate. 

Let $A_k^\ell$ be the bundle of player $k$ at the beginning of round $\ell$, and let $B_k^\ell$ denote the bundle of player $k$ just before EliminateEnvyCycles is run on round $\ell$. Let $A^\ell = (A_1^\ell...A_n^\ell)$ and $B^\ell = (B_1^\ell...B_n^\ell)$. In this proof, we use $k$ and $k'$ to denote a generic player; $i$ and $j$ refer exclusively to the variables in the while-loop.

We proceed by induction on $\ell$. Initially, all players have empty bundles, which trivially satisfies $\frac{1}{2}$-EFX. Thus assume the partial allocation at the beginning of round $\ell$ is $\frac{1}{2}$-EFX. We will show that the partial allocation at the beginning of round $\ell+1$ is $\frac{1}{2}$-EFX. The partial allocation at the beginning of round $\ell+1$ $A^{\ell+1}$ is equal to EliminateEnvyCycles$(B^\ell)$. Thus by Lemma~\ref{lem:remove-cycles}, it suffices to show that $B^\ell$ is $\frac{1}{2}$-EFX.

If the body of the if-statement (lines 10-12) is not executed, the allocation $B^\ell$ is $\frac{1}{2}$-EFX by definition. Thus assume the body of the if-statement is executed. Then $B_j^\ell = A_j^\ell$, because $g^*$ was added and then removed. Thus for all $k\neq i$, $B_k^\ell = A_k^\ell$. 

We say that a pair $(k, k')$ is $\frac{1}{2}$-EFX in $B^\ell$ if $v(B_{k}^\ell) \geq \frac{1}{2} v(B_{k'}\backslash\{g\})$ for all $g\in B_{k'}^\ell$. We know that $A^\ell$ is $\frac{1}{2}$-EFX by assumption. Therefore since $B_k^\ell = A_k^\ell$ for all $k\neq i$, all pairs $(k, k')$ where $k \neq i$ and $k'\neq i$ remain $\frac{1}{2}$-EFX in $B^\ell$. Furthermore, since $B_i^\ell = \{g^*\}$, the pair $(k,i)$ is $\frac{1}{2}$-EFX for all players $k$, since $B_i^\ell \backslash\{g\} = \emptyset$ for all $g \in B_i^\ell$. 

It remains only to show that the pairs $(i,k)$ are $\frac{1}{2}$-EFX for all players $k$. We do this by showing that $v_i(B_i^\ell) > v_i(A_i^\ell)$. The fact that this inequality is strict will be important later in showing that the algorithm terminates.

We know that $j$ was unenvied at the beginning of round $\ell$, so $v_i(A_i^\ell) \geq v_i(A_j^\ell)$. Since the body of the if-statement executed, we also know that there exists $g \in A_j^\ell \cup \{g^*\}$ such that $v_i(A_i^\ell) < \frac{1}{2} v_i(A_j^\ell \cup \{g^*\} \backslash \{g\})$. Thus $v_i(A_i^\ell) < \frac{1}{2} v_i(A_j^\ell \cup \{g^*\})$, which will be all we need. Therefore,
\begin{align}
v_i(A_i^\ell) <&\ \mfrac{1}{2} v_i(A_j^\ell \cup \{g^*\})\label{line1}\\
\leq&\ \mfrac{1}{2}(v_i(A_j^\ell) + v_i(\{g^*\}))\label{line2}\\
\leq&\ \mfrac{1}{2}(v_i(A_i^\ell) + v_i(\{g^*\}))\label{line3}
\end{align}
where \ref{line2} follows from \ref{line1} due to $v_i$ being subadditive, and \ref{line3} follows from \ref{line2} due to $v_i(A_i^\ell) \geq v_i(A_j^\ell)$. Therefore,
\begin{align*}
v_i(A_i^\ell) -  \frac{1}{2}v_i(A_i^\ell)<&\ \frac{1}{2} v_i(\{g^*\}) \\
v_i(A_i^\ell) <&\ v_i(\{g^*\}) \\
v_i(A_i^\ell) <&\ v_i(B_i^\ell)
\end{align*}

Consider an arbitrary player $k\neq i$. Since $A^\ell$ is $\frac{1}{2}$-EFX, we have $v_i(A_i^\ell) \geq \frac{1}{2} v_i(A_k^\ell \backslash \{g\})$ for all $g\in A_k^\ell$. Since $v_i(B_i^\ell) > v_i(A_i^\ell)$ and $B_k^\ell = A_k^\ell$ for all $k\neq i$, we have $v_i(B_i^\ell) \geq \frac{1}{2} v_i(B_k^\ell \backslash\{g\})$ for all $g\in B_k^\ell$ as well. Therefore the pair $(i,k)$ is $\frac{1}{2}$-EFX for all players $k$. 

Thus every pair of players is $\frac{1}{2}$-EFX in $B^\ell$, so $B^\ell$ is $\frac{1}{2}$-EFX. This shows that the partial allocation at the end of each round is $\frac{1}{2}$-EFX, and so any allocation returned by the algorithm is $\frac{1}{2}$-EFX.

It remains to show that Algorithm~\ref{alg:apx}
terminates. We use a potential function argument.
For round~$\ell$, define 
\[
\phi(\ell) = \sum_{k=1}^n v(A_k^\ell). 
\]
We noted above that if round $\ell$ falls under Case 2, only $i$'s bundle changes, and we have the strict inequality $v_i(B_i^\ell) > v_i(A_i^\ell)$. Therefore $v_i(A_i^{\ell+1}) > v_i(A_i^\ell)$. Thus if round $\ell$ falls under Case 2, we have $\phi(\ell + 1) - \phi(\ell) > 0$.

If round $\ell$ falls under Case 1, only $j$'s bundle changes, and we have $v_j(A_j^{\ell+1}) \geq v_i(A_i^\ell)$. Therefore if round $\ell$ falls under Case 1, we have $\phi(\ell + 1) - \phi(\ell) \geq 0$.

In any round which falls under Case 1, $|P|$ decreases by
one. Therefore if $m$ rounds pass without Case 2 occurring, $P$ becomes empty, and the algorithm terminates. Thus while the algorithm has not terminated, Case 2 must occur at
least once every $m$ rounds, and so $\phi(\ell + m) - \phi(\ell) > 0$ for
all $\ell$.

The number of possible partial allocations is at most $(n+1)^m$: each
good can be given to one of the $n$ players, or left in the
unallocated pool. 
Thus the number of distinct values $\phi$ can take on is at most
$(n+1)^m$, and so $\phi$ can increase at most that many times. Thus
after $m (n+1)^m$ rounds, the algorithm must have
terminated. 
\end{proof}

Finally, we briefly show that $\frac{1}{2}$-EFX and EF1 are incomparable, meaning that neither property implies the other. Recall that an allocation $A$ is EF1 if for all $i,j$ where $A_j \neq \emptyset$, there exists $g \in A_j$ where $v_i(A_i) \geq v_i(A_j \backslash\{g\})$.

Consider the additive valuations on the left, and let $A = (\{a,b\}, \{c\})$. $A$ is EF1 because $v_2(A_2) \geq v_2(A_1 \backslash \{a\})$, but $A$ is not $\frac{1}{2}$-EFX because $v_2(A_2) < \frac{1}{2}v_2(A_1 \backslash \{b\})$.

Now consider the valuations on the right, and let $A = (\{a,b,c\}, \{d\})$. Then $A$ is not EF1, because $v_2(A_2) < v_2(A_1 \backslash \{g\})$ for all $g\in A_1$, but $A$ is $\frac{1}{2}$-EFX, because $v_2(A_2) \geq \frac{1}{2}v_2(A_1 \backslash \{g\})$ for all $g \in A_1$.

\begin{center}
\begin{tabular}{ c|ccc} 
& a & b & c\\
\hline
player 1 & 3 & 1 & 0 \\
player 2 & 3 & 0 & 1\\
\end{tabular}
\hspace{.12 in}
\begin{tabular}{ c|cccc} 
& a & b & c & d\\
\hline
player 1 & 1 & 1 & 1 & 1 \\
player 2 & 1 & 1 & 1 & 1\\
\end{tabular}
\end{center}

%% file: conclusion.tex
\section{Conclusion and future work}\label{sec:conclusion}

In this paper, we provided the first general results on the fairness
concept of envy-freeness up to any good. Our most technically involved result was an exponential lower bound on the number of queries required by any deterministic algorithm to find an EFX allocation, via a reduction from local search. To complete the lower bound, we proved an exponential lower bound on the number of queries required to find a local maximum on $K(2k+1, k)$. We used results from Dinh and Russell~\shortcite{Dinh10:Quantum} and Valencia-Pabon and Vera~\shortcite{Valencia-Pabon05:Diameter} to obtain an exponential lower
bound for randomized algorithms as well. Our EFX lower bounds hold even
for two players with identical submodular valuations.

Next, we showed that for $n$
players with general but identical valuations, a modification of the leximin solution is guaranteed to be EFX. We showed how this result can be adapted into a cut-and-choose
protocol for finding an EFX allocation between two players with
general and possibly distinct valuations.

We also considered satisfying EFX and Pareto optimality together. We showed that if players are allowed to have zero value for a good being added to their bundle, it is impossible to guarantee EFX and Pareto optimality simultaneously. However, if we assume that a player's value for her bundle is strictly increased by adding any good (even just by some tiny $\epsilon$), the leximin solution is EFX and PO two settings: for $n$ players with general but identical valuations, and for two players with possibly distinct additive valuations. We view the latter result as our result of most practical significance: assuming nonzero marginal utility, it provides stronger guarantees the currently deployed algorithm on Spliddit, even in simple examples. Our other significant positive result was an algorithm for finding a $\frac{1}{2}$-EFX allocation for any number of players with subadditive valuations.

The ideal next step would be to consider EFX with distinct valuations and more
than two players. This problem seems quite challenging, even for the
special case of additive valuations. Indeed, Caragiannis et
al.~\shortcite{Caragiannis16:Nash} were unable to settle the question of whether
EFX allocations in that context always exist, ``despite significant
effort.'' 
The problem seems highly non-trivial even for three players with
different additive valuations. We suspect that at least for general valuations, there exist instances where no EFX allocation exists, and it may be easier to find a counterexample in that setting. Similarly, finding a counterexample to EFX and Pareto optimality together for additive valuations and more than two players (assuming nonzero marginal utility) is another avenue that may be more tractable.

Another direction is to pursue stronger lower bounds for finding an EFX allocation. In particular, communication complexity allows players unlimited computation and queries, and only measures the number of bits transmitted. The cut-and-choose protocol from Section~\ref{sec:gen-id} constitutes a linear communication protocol for two players with general and possibly distinct valuations to compute an EFX allocation, so any communication complexity lower bound would need to consider more than two players. On the other hand, we know finding an EFX allocation to be hard in the query model even for two players, which suggests an interesting separation. 

More generally, communication complexity is one example of a topic that has been studied in algorithmic mechanism design and may be useful in the study of fair division. Another such topic is the hierarchy of complement-free valuations (additive, submodular, subadditive, etc.). Our work already implies separations between these valuation classes from a fair division perspective, and suggests that fair division with different classes of player valuations deserves further study.

%% file: additionalProofs.tex
\section{Additional proofs}\label{sec:add-proofs}

\begin{proof}[Proof of Theorem~\ref{thm:total}]
To show that $\prec_{++}$, we need to show that for any allocations $A, B,$ and $C$, $A \prec_{++} A$ is false, and that ($A \prec_{++} B$ and $B \prec_{++} C$) implies $A \prec_{++} C$.

We first show that $A\prec_{++} A$ is false. The key fact is that for a given allocation $A$, there is only one possible ordering of the players $X^A$: were this not true, $\prec_{++}$ could fail to produce a total ordering.\footnote{Consider two players with identical valuations and one good $a$, where $v(\{a\}) = 0$. Let $A = (\emptyset, \{a\})$. Suppose both (1,2) and (2,1) are valid orderings of the players according to $A$, and suppose we run $A \prec_{++} A$ with the left hand side $A$ using the ordering (1,2) and the right hand side $A$ using (2,1). Then at $\ell = 1$, $\emptyset$ from the left hand side $A$ will be compared with $\{a\}$ from the right hand side $A$, and $A \prec_{++} A$ will return true.} Therefore on each iteration, the same player is considered from each copy of $A$. Thus on each iteration, the two bundles compared will be the same, so $A\prec_{++}A$ never terminates until it passes through all $\ell \in [n]$ and returns false at the very end.

It remains to show that ($A \prec_{++} B$ and $B \prec_{++} C$) implies $A \prec_{++} C$. Suppose $A \prec_{++} B$ and $B \prec_{++} C$. Let $\ell_1$, $\ell_2$, and $\ell_3$ be the iterations on which $A\prec_{++} B$, $B \prec_{++} C$ and $A\prec_{++} C$ terminate, respectively. For $x\in \{1,2,3\}$, let $i_x = X_{\ell_x}^A$, $j_x = X_{\ell_x}^B$, and $k_x = X_{\ell_x}^C$.

Since $A\prec_{++} B$ terminates on iteration $\ell_1$, we have $v(A_{i_1}) < v(B_{j_1})$ or $|A_{i_1}| < |B_{j_1}|$. Similarly, since $B\prec_{++} C$ terminates on iteration $\ell_2$, we have $v(B_{i_2}) < v(C_{j_2})$ or $|B_{i_2}| < |C_{j_2}|$.

First we argue that $\ell_3 \geq \min(\ell_1, \ell_2)$. Suppose $\ell < \min(\ell_1, \ell_2)$: then $A \prec_{++} B$ and $B\prec_{++} C$ do not terminate until after iteration $\ell_3$. Therefore $v(A_{i_3}) = v(B_{j_3})$, $|A_{i_3}| = |B_{j_3}|$, $v(B_{j_3}) = v(C_{k_3})$, and $|B_{j_3}| = |C_{k_3}|$. Therefore $v(A_{i_3}) = v(C_{k_3})$ and $|A_{i_3}| = |C_{k_3}|$, so $A \prec_{++} C$ could not have terminated on iteration $\ell_3$, which is a contradiction. Therefore $\ell_3 \geq \min(\ell_1, \ell_2)$. We proceed by case analysis.

Case 1: $\ell_1 < \ell_2$. Since $B \prec_{++} C$ did not terminate until after iteration $\ell_1$, we have $v(B_{j_1}) = v(C_{k_1})$ and $|B_{j_1}| = |C_{k_1}|$. Therefore $v(A_{i_1}) < v(C_{k_1})$ or $|A_{i_1}| < |C_{k_1}|$. We know that $A \prec_{++} C$ cannot have terminated prior to $\ell_1$, since $\ell_3 \geq \min(\ell_1, \ell_2) = \ell_1$. Therefore $A \prec_{++} C$ will terminate on iteration $\ell_1$ and return true, so $A \prec_{++} C$ holds in Case 1.

Case 2: $\ell_2 < \ell_1$. This case is similar. Since $A \prec_{++} B$ did not terminate until after iteration $\ell_2$, we have $v(A_{i_2}) = v(B_{j_2})$ and $|A_{i_2}| = |B_{j_2}|$. Therefore $v(A_{i_2}) < v(C_{k_2})$ or $|A_{i_2}| < |C_{k_2}|$. We know that $A \prec_{++} C$ cannot have terminated prior to $\ell_2$, since $\ell_3 \geq \min(\ell_1, \ell_2) = \ell_2$. Therefore $A \prec_{++} C$ will terminate on iteration $\ell_2$ and return true, so $A \prec_{++} C$ holds Case 2.

Case 3: $\ell_1 = \ell_2$. In this case we have $i_1 = i_2$, $j_1 = j_2$, and $k_1 = k_2$. Therefore
\begin{align*}
v(A_{i_1}) < v(B_{j_1})\ \textnormal{or}&\ \ \big(v(A_{i_1}) = v(B_{j_1})\ \textnormal{and}\ |A_{i_1}| < |B_{j_1}|\big),\ \textnormal{and}\\
v(B_{j_1}) < v(C_{k_1})\ \textnormal{or}&\ \ \big(v(B_{j_1}) = v(C_{k_1})\ \textnormal{and}\ |B_{j_1}| < |C_{k_1}|\big)
\end{align*}
Note that $v(A_{i_1}) \leq v(B_{j_1})$ and $v(B_{j_1}) \leq v(C_{k_1})$. Therefore if either $v(A_{i_1}) < v(B_{j_1})$ or $v(B_{j_1}) < v(C_{k_1})$, we have $v(A_{i_1}) < v(C_{k_1})$. We know $A\prec_{++} C$ cannot have terminated before $\ell_1=\ell_2$ since $\ell_3 \geq \min(\ell_1, \ell_2)$, so if $v(A_{i_1}) < v(C_{k_1})$, $A \prec_{++} C$ terminates on iteration $\ell_1$ and returns true.


Thus assume $v(A_{i_1}) = v(B_{j_1})$ and $v(B_{j_1}) = v(C_{k_1})$: then $|A_{i_1}| < |B_{j_1}|$ and $|B_{j_1}| < |C_{k_1}|$. Therefore $v(A_{i_1}) = v(C_{k_1})$ and $|A_{i_1}| < |C_{k_1}|$, so $A\prec_{++} C$ terminates on iteration $\ell_1$ and returns true. Therefore $A \prec_{++} C$ in Case 3. This shows that ($A \prec_{++} B$ and $B \prec_{++} C$) implies $A \prec_{++} C$, and completes the proof.
\end{proof}

\begin{proof}[Proof of Theorem~\ref{thm:yes-efx-po-add-id}]
Let $Z$ be the set of all goods $g$ where $v(\{g\}) = 0$. Therefore for all $g \in M \backslash \{Z\}$, we have $v(\{g\}) > 0$, so $v$ has nonzero marginal utility over the set of goods $M \backslash \{Z\}$.

Let $A = (A_1...A_n)$ be the leximin allocation over $M \backslash \{Z\}$. By Theorem~\ref{thm:gen-id-po}, $A$ is EFX and PO over $M \backslash \{Z\}$.

Let $i$ be the minimum utility player in $A$. Define a new allocation $B$ over all of $M$ where $B_i = A_i \cup \{Z\}$ and $B_{j} = A_j$ for all $j\neq i$. Since $v(Z) = 0$, we have $v(B_j) = v(A_j)$ for all $j$. Therefore since $i$ had minimum utility in $A$, $i$ also has minimum utility in $B$.

To see that $B$ is EFX, consider arbitrary players $j$ and $k$, and any $g \in B_k$. If $i \neq k$, we have $A_k = B_k$. Since $A$ is EFX, we have $v(B_j) = v(A_j) \geq v(A_k \backslash\{g\}) = v(B_k \backslash\{g\})$. If $i = k$, then $v(B_j) \geq v(B_k) \geq v(B_k \backslash\{g\})$, since $i$ has minimum utility in $B$. This shows that $B$ is EFX.

To see that $B$ is PO, observe that the way the goods in $Z$ are allocated has no effect on the values of the bundles. Therefore the goods in $Z$ have no effect on the Pareto optimality of the allocation, so the Pareto optimality of $B$ follows directly from the Pareto optimality of $A$.
\end{proof}


%% file: algoSameRanking.tex
\section{A setting where an EFX allocation can be computed quickly}\label{sec:same-ranking}

Finally, we describe a setting in which an EFX allocation always
exists and can be computed in polynomial time (counting both the value
queries and all additional computation done by an algorithm). Our result will hold when players have {\em additive valuations with
identical rankings}, meaning that
all players agree on the relative ordering of individual goods.  This
is, for all players $i$ and $j$, and for all goods $g_1$ and $g_2$,
$v_i(g_1) \geq v_i(g_2)$ whenever $v_j(g_1) \geq v_j(g_2)$. This will also yield a polynomial time algorithm for computing an EFX allocation for two players with additive (possibly distinct) valuations.

Requiring identical rankings is not as strong as requiring identical valuations. For example, let
$v_1(g_1) = 1, v_1(g_2) = 2, v_1(g_3) = 4$ and
$v_2(g_1) = 2, v_2(g_2) = 3, v_2(g_3) = 4$. Then the rankings are
identical, but $v_1(\{g_1, g_2\}) < v_1(g_3)$, whereas
$v_2(\{g_1, g_2\}) > v_2(g_3)$.

While strong,
there are certainly real-world contexts where this assumption makes
sense.
For example, if the goods are apartments (with differing square
footage), airline tickets (with differing numbers of stops and classes
of service), or baseball pitchers (with differing statistics), it is
plausible that buyers
generally agree on which goods are more valuable than others, but
disagree on the exact values of these goods.

Our algorithm (Algorithm~\ref{alg:additive-same-ranking}) is reminiscent of our algorithm for finding a $\frac{1}{2}$-EFX allocation for any number of players with subadditive valuations from Section~\ref{sec:apx-algo}, in that we allocate the goods in rounds and ensure that the envy graph is acyclic at the beginning of each round. However, here we never return goods to the pool, and allocate the goods in descending order of value.

\begin{algorithm}[tb]
\centering
\begin{algorithmic}[1]
  \Function{GetEFXAllocationSameRanking}{$n, m, (v_1...v_n)$}
  \State $P \gets \text{Sorted}([m])$ \Comment{Sort in descending order: $P_1 = \max(P)$}
  \ForAll{$i \in [n]$}
  \State $A_i \gets \emptyset$
  \EndFor
  \ForAll{$i \in [m]$}
  \State $j \gets \text{FindUnenviedPlayer}(A_1, A_2...A_n)$
  \State $A_j \gets A_j \cup \{P_i\}$
  \State $(A_1, A_2...A_n) \gets \text{EliminateEnvyCycles}(A_1, A_2...A_n)$
  \EndFor
\Return $(A_1, A_2...A_n)$
\EndFunction
\end{algorithmic}
\caption{Find an EFX allocation for additive valuations with identical ranking}
\label{alg:additive-same-ranking}
\end{algorithm}

Recall that Lemma~\ref{lem:remove-cycles} gives a process that can be used to ensure the envy graph is acyclic: if an envy cycle exists, bundles can be permuted along this cycle such the number of edges in the envy graph decreases by at least one. The function EliminateEnvyCycles repeatedly performs this process until the envy graph is acyclic.

%
%

\begin{theorem}\label{thm:additive-same-ranking}
For additive valuations with identical rankings, Algorithm~\ref{alg:additive-same-ranking} terminates with an EFX allocation in $O(mn^3)$ time.
\end{theorem}

\begin{proof}
We first argue that at all times, $v_i(A_j) - v_i(A_i) \leq v_i(g^*)$ where $g^*$ is the good most recently added to what is currently $A_j$. Since bundles may have been permuted by EliminateEnvyCycles, $j$ may not have been in possession of what is currently $A_j$ at the time $g^*$ was added. This does not affect the proof, however: it is sufficient to interpret $A_j$ as ``the bundle that currently belongs to $j$". Thus instead of saying  ``$i$ did not envy $j$ at the time", we will say ``$i$ did not envy $A_j$ at the time".

Observe that a good is only allocated to a player whom no one envies. Thus directly before $g^*$ was added to $A_j$, $i$ did not envy $A_j$: at that point $v_i(A_j) - v_i(A_i) \leq 0$. Therefore directly after $g^*$ was given to $j$, $v_i(A_j) - v_i(A_i) \leq v_i(g^*)$. Since $v_i(A_i)$ can only have grown since then, we have $v_i(A_j) - v_i(A_i) \leq v_i(g^*)$ until a new good is added to $A_j$.

Since the goods are allocated in decreasing order of value, the good most recently added to $A_j$ must also be the least valuable good in $A_j$. Therefore at all times, $v_i(A_j) - v_i(A_i) \leq \min\limits_{g\in A_j} v_i(g)$, and so $v_i(A_i) \geq v_i(A_j) - \min\limits_{g\in A_j} v_i(g)$. For additive valuations, this is equivalent to $v_i(A_i) \geq v_i(A_j \backslash \{g\})$ for all $g\in A_j$. Therefore the allocation at all times is EFX, so the final allocation is EFX.

Finally, we show that Algorithm~\ref{alg:additive-same-ranking}
terminates in $O(mn^3)$ time. Each time a good is allocated, any edges
added to the envy graph must point to the recipient. Thus at most $n$
edges are added to the envy graph on each round, and so at most $mn$
edges are added to the graph over the course of the algorithm. Each
time a cycle is detected and bundles are permuted along that cycle using Lemma~\ref{lem:remove-cycles}, at least one edge is removed from the graph. Therefore this process is performed
at most $mn$ times. Each time this process is performed, we may have to compute a
large part of the envy graph, which can take $O(n^2)$ time. 
Thus the overall running time bound is $O(mn^3)$.
\end{proof}

This algorithm is easily generalizable to
general valuations under the condition that all players agree on a
single ordering of the marginal values of the goods. Specifically,
there must be an ordering of the goods $(g_1, g_2,\ldots,g_m)$ where for
any set $S$, any player $i$, and all $j$, we have
$v_i(S\cup \{g_j\}) \geq v_i(S\cup\{g_{j+1}\})$. This ordering must be
fixed across all sets $S$. Then instead of allocating goods in
descending order of value, we allocate goods in descending order of
marginal value, and the analog of Theorem~\ref{thm:additive-same-ranking} holds, with
essentially the same proof.

Finally, we note that Algorithm~\ref{alg:additive-same-ranking} can be used to compute an EFX allocation for two players with additive (possibly distinct) valuations in polynomial time. We use a cut-and-choose argument similar to that of Theorem~\ref{thm:two-players}: player 1 runs Algorithm~\ref{alg:additive-same-ranking} with two copies of herself to find an allocation which will be EFX from her viewpoint, regardless of which bundle she receives. Then player 2 chooses her favorite bundle in the resulting allocation, so the allocation will be fully envy-free from her viewpoint.